\documentclass[11pt,oneside]{amsart}

\usepackage{color}
\usepackage{graphicx,mathrsfs,phonetic,amsmath,stmaryrd,abraces,mathtools,mathabx}
\usepackage{fullpage}

\usepackage[bookmarks=true,%
    colorlinks=true,%
    linkcolor=blue,%
   citecolor=blue,%
   filecolor=blue,%
    menucolor=blue,%
    urlcolor=blue,%
    breaklinks=true]{hyperref}

\calclayout

\newcommand{\dd}{\mathrm{d}}

\def\Res{\mathop{\mathrm{Res}}}

\def\d{\partial}

\newcommand{\beq}{\begin{equation}}
\newcommand{\eeq}{\end{equation}}
\newcommand{\bea}{\begin{eqnarray}}
\newcommand{\eea}{\end{eqnarray}}

\newtheorem{theorem}{Theorem}[section]

\newtheorem{proposition}[theorem]{Proposition}
\newtheorem{lemma}[theorem]{Lemma}
\newtheorem{corollary}[theorem]{Corollary}

\theoremstyle{remark}

\newtheorem{example}{Example}[section]

\newtheorem{remark}[example]{Remark}

\newtheorem{assumption}[example]{Assumption}

\begin{document}

\title[]{From topological recursion to wave functions and PDEs quantizing hyperelliptic curves}
\author{Bertrand Eynard}
\address{Institut de Physique Th\'{e}orique de Saclay and Institute des Hautes \'{E}tudes Scientifiques}
\email{bertrand.eynard@ipht.fr}
\author{Elba Garcia-Failde}
\address{Institut de Physique Th\'{e}orique de Saclay and Institute des Hautes \'{E}tudes Scientifiques}
\email{garciafailde@ihes.fr}

\begin{abstract}
Starting from loop equations, we prove that the wave functions constructed from topological recursion on families of degree $2$ spectral curves with a global involution satisfy a system of partial differential equations, whose equations can be seen as quantizations of the original spectral curves. The families of spectral curves can be parametrized with the so-called times, defined as periods on second type cycles, and with the poles. These equations can be used to prove that the WKB solution of many isomonodromic systems coincides with the topological recursion wave function, which proves that the topological recursion wave function is annihilated by a quantum curve. This recovers many known quantum curves for genus zero spectral curves and generalizes this construction to hyperelliptic curves.
\end{abstract}

\maketitle

\tableofcontents


\section{Introduction} 

Topological recursion is a powerful tool, which was first discovered in the context of large size asymptotic expansions in random matrix theory \cite{E1MM,CEO06,CE061,CE062} and established as an independent universal theory around 2007 \cite{EO07}. Its most important role was to unveil a common structure in many different topics in mathematics and physics, which helped building bridges among them and gaining general context. For instance, it has been related to fundamental structures in enumerative geometry and integrable systems, such as intersection theory on moduli spaces of curves and cohomological field theories.

A quantum curve is a Schr\"{o}dinger operator-like non-commutative analogue of a plane curve that annihilates the so-called wave function, which can be seen as a WKB asymptotic solution of the corresponding differential equation. Inspired by the intuition coming from matrix models, it has been conjectured that there exists such a quantum curve associated to a spectral curve, which is the input of the topological recursion, and whose WKB asymptotic solution is reconstructed by the topological recursion output.

This claim was verified in \cite{BouchardEynard_QC} for a class of spectral curves called admissible, which are basically spectral curves whose Newton polygons have no interior point. Admissible spectral curves include a very large number of spectral curves of genus $0$. Therefore they recover many cases previously studied in the literature in various algebro-geometric contexts. In the present work, we go beyond admissible spectral curves and study the quantum curve problem for spectral curves with a global involution, given by algebraic curves whose defining polynomials are of the form $y^2=R(x)$, with $R$ a rational function on $x$. This setting includes the genus $0$ spectral curves with a global involution $y\mapsto -y$, such as the well-known Airy curve, which are many less that the set of admissible curves, but it also includes all genus $1$ spectral curves, i.e.~all elliptic curves, and all hyperelliptic curves, which are curves of genus $\hat{g}>1$ where $R$ is a polynomial in $x$.

Quantum curves encode enumerative invariants in an interesting way, and help building the bridge between the geometry of integrable systems and topological recursion. One of the most celebrated applications of quantum curves is in the context of knot theory, where the quantum curve of the $A$-polynomial of a knot provides a conjectural constructive generalization of the volume conjecture \cite{DijkFuji,DijkFujiMana, BEApol}.

\subsection{Quantum curves and topological recursion} We start by presenting the idea of quantum curves and their relation to topological recursion. Consider 
$P\in \mathbb{C}[x,y]$ and let
$$
\mathcal{C}=\{(x,y)\in\mathbb{C}^2\mid P(x,y)=0\}
$$ be a plane curve. 

Consider $\hbar>0$ a formal parameter. A quantization of the plane curve $\mathcal C$,  is a differential operator $\widehat{P}$ of the form
$$
\widehat{P}(\widehat{x},\widehat{y};\hbar)=P_0(\widehat{x},\widehat{y}) +O(\hbar),
$$
where $\widehat{x}=x\cdot\,$, $\widehat{y}=\hbar\frac{\dd}{\dd x}$. In fact, $\widehat P$ is a normal-ordered operator valued (all the $\widehat{y}$ in a monomial are placed to the right of all the $\widehat{x}$) formal series (or transseries) whose leading order term $P_0(\widehat{x},\widehat{y})$ recovers the polynomial equation of the original spectral curve. Actually, in general, $P_0(x,y)$ can be a reducible polynomial with $P(x,y)$ one of its factors.
The operators $\widehat{x}$ and $\widehat{y}$ satisfy the following commutation relation which justifies the name ``quantization'':
$$
[\widehat{y},\widehat{x}]=\hbar.
$$
One can consider a Schr\"{o}dinger-type differential equation
\begin{equation}\label{quantumC}
\widehat{P}(\widehat{x},\widehat{y})\psi(z,\hbar)=0, \text{ with } z\in \mathcal{C},
\end{equation}
whose solution can be calculated via the WKB method, that is we require $\psi$ to have a formal series (resp.~transseries) $\hbar$ expansion of the form
$$
\psi(z,\hbar)=\exp\bigg(\sum_{m\geq 0}\hbar^{m-1}S_m(z)\bigg)
$$
(resp.~of the form of a formal series in powers of exponentials of inverse powers of $\hbar$ whose coefficients are formal power series in $\hbar$).
The coefficients $S_k(z)$ are determined recursively via \eqref{quantumC}. One fundamental question is if the formal solution $\psi$ can be computed directly from the original plane curve $\mathcal{C}$. The conjectural answer is provided by the topological recursion and is already established in many cases.

The input data of the topological recursion is called spectral curve. For the purpose of this paper, a \emph{spectral curve} will be given as in \cite{GIS}, i.e.~by the data $(\Sigma,x,y, \omega_{0,1},\omega_{0,2})$, where $\Sigma$ is a compact Riemann surface, and $x$ and $y$ are meromorphic functions on $\Sigma$ such that the zeroes of $\dd x$ do not coincide with the zeroes of $\dd y$. Then $x$ and $y$ must be algebraically dependent, i.e.~$P(x,y)=0$ with $P\in\mathbb{C}[x,y]$. We consider $\omega_{0,1}\coloneqq y\dd x$ and $\omega_{0,2}$ a symmetric bi-differential $B$ on $\Sigma^2$ with double poles along the diagonal and vanishing residues. The output of the topological recursion are symmetric meromorphic multi-differentials $\omega_{g,n}(z_1,\ldots,z_n)$ on $\Sigma^n$. In particular, for $n=0$, the $F_g\coloneqq \omega_{g,0}$ are complex numbers or expressions depending on parameters of the spectral curves. In the present work, we assume that the poles of $y$ but not of $x$ cannot be ramification points of $x$.

The perturbative wave function $\psi(z)$ constructed from topological recursion is defined (see \cite{EO07}) as 
$$
\frac{1}{x(z)-x(o)}\exp\bigg(\sum_{g\geq 0, n\geq 1}\frac{\hbar^{2g-2+n}}{n!}\int_o^z\cdots\int_o^z \Big(\omega_{g,n}(z_1,\ldots,z_n)-\delta_{g,0}\delta_{n,2}\frac{\dd x(z_1)\dd x(z_2)}{(x(z_1)-x(z_2))^2}\Big) \bigg),
$$
where $o\in \Sigma$ is a chosen base-point for integration. 
In general the quantum curve is obtained by choosing $o=\zeta$ such that $x(\zeta)=\infty$, and one may need to regularize the $(g,n)=(0,1)$ term in the limit $o\to \zeta$. We choose a  regularization for $(g,n)=(0,2)$ which slightly differs from some part of the literature and produces the first factor of the expression. We emphasize that our definition transforms as a spinor $\frac12$-form under change of coordinates. Actually, we will generalize the definition of the wave function by allowing integration over any divisor as in \cite{GIS}.

A further question is if the differential operator $\widehat{P}$ can be directly constructed from the topological recursion. This has also been answered affirmatively in many cases. In this article, we actually construct an operator that we believe is a more fundamental object, which appears more naturally for curves of genus $\hat{g}>0$, and provides a PDE which also allows to reconstruct the wave function~$\psi$. Our system of PDEs will also imply a quantum curve in the more classical sense considered in this section.

\subsection{Generalized cycles} Let us now recall the concept of generalized cycles on a Riemann surface as in \cite{GIS}, which will help us introduce suitable local coordinates in the space of spectral curves. 
These local coordinates can be seen as deformation parameters giving rise to families of spectral curves and will play a key role when producing our system of PDEs for a large class of spectral curves.

The so-called times $t_i$, introduced in \cite{GIS}, can be viewed as local coordinates in the space of spectral curves. Time deformations $\partial_{t_i}$ belong to the tangent space, which is isomorphic to the space of meromorphic differential forms on the spectral curve, and via form-cycle duality it can be identified with the space of generalized cycles on the spectral curve.
In \cite{GIS}, generalized cycles are defined as elements of the dual of the space of meromorphic forms on $\Sigma$ such that integrating $\omega_{0,2}=B$ on them gives meromorphic $1$-forms.

In practice a generating family is given by three kinds of cycles, dual to three kinds of forms:
\begin{itemize}
\item \textbf{1st kind cycles:} This type of cycles are usual non-contractible cycles, i.e.~elements $[\gamma]\in \textrm{H}_1(\Sigma,\mathbb{C})$. If $\Sigma$ is compact of genus $\hat{g}$, then $\dim \textrm{H}_1(\Sigma,\mathbb{C})=2 \hat{g}$.
\item \textbf{2nd kind cycles:} A cycle of second type $\gamma=\gamma_p . f$ consists of a small circle $\gamma_p$ around a point $p\in \Sigma$ weighted by a function $f$ holomorphic in a neighborhood of $\gamma_p$ and meromorphic in a neighborhood of $p$, with a possible pole at $p$ (of any degree), i.e.~by definition $\int_\gamma \omega \coloneq 2\pi i\Res_{p} f \omega$.
\item \textbf{3rd kind cycles:} They are open chains $\gamma=\gamma_{q\rightarrow p}$ (paths up to homotopic deformation with fixed endpoints), whose boundaries $\partial\gamma=[p]-[q]$ are degree zero divisors.
\end{itemize}

Let $x : \Sigma \rightarrow \mathbb{C}$ be the meromorphic function that makes the spectral curve a branched cover of the Riemann sphere. A basis of functions which are meromorphic in a neighborhood of $p\in\Sigma$ is given by
$$
\{\xi_p^k\}_{k\in\mathbb{Z}}, \text{ with } \xi_p=(x-x(p))^{1/\text{ord}_p(x)}.
$$
If $x(p)=\infty$, we set $\xi_p=x^{1/\text{ord}_p(x)}$, with $\text{ord}_p (x)<0$.
The following set of cycles generates an integer lattice in the space of second kind cycles:
\begin{align}\label{ABcycles}
\mathcal{A}_{p,k} & =\gamma_p.\xi_p^k, \, \, p\in\Sigma, k\geq 0, \\
\mathcal{B}_{p,k} & =\frac{1}{2\pi i}\gamma_p.\frac{\xi_p^{-k}}{k}, \, \, p\in\Sigma, k\geq 1. \label{2ndTypeCycles}
\end{align}

Given a meromorphic $1$-form $\omega$ on $\Sigma$, for every pole $p$ of $\omega$, we define for every $j\geq 0$ the \emph{KP times}:
\begin{equation}
t_{p,j}\coloneqq  \frac{1}{2\pi i}\int_{\mathcal{A}_{p,j}}\omega = \frac{1}{2\pi i}\int_{\gamma_{p}}\xi_p^j \,\omega = \underset{p}{\mathrm{Res}}\, (\xi_p)^j \omega\,,
\end{equation}
so that 
$$
\omega \sim \sum_{j=0}^{\text{deg}_p(\omega)} t_{p,j} \xi_p^{-j-1}d\xi_p + \text{analytic at }p.
$$
Since we assumed $\Sigma$ to be compact, the number of poles is finite. Moreover, all the times with $j\geq \mathrm{deg}_p \omega$ are vanishing. Therefore, only a finite number of times are non-zero.
\subsection{Context and outline} One of our motivations to study the problem of quantum curves for any spectral curve with a global involution was to be able to recover the whole isomonodromic system associated to Painlev\'{e} I just from loop equations. We also aimed to give the first quantum curves for spectral curves of genus $\hat{g}>1$ and we were especially interested to see how introducing deformations with respect to the times could give rise to systems of PDEs that we consider more natural in general.

\subsubsection{Comparison to the literature}

In \cite{Bouchard_Quantizing}, they generalize the techniques employed in \cite{BouchardEynard_QC} to find the quantum curves for admissible curves to apply them to the family of genus one spectral curves given by the Weierstrass equation. They find an order two differential operator that annihilates the perturbative wave function $\psi$. However, it is not a quantum curve, since it contains infinitely many $\hbar$ corrections which are not meromorphic functions of $x$. They also check the first orders of the conjectural quantum curve \cite{Eynard_2009, EMhol, BEInt} for the non-perturbative wave function.

In \cite{IwakiSaenz}, they focus on the Painlev\'{e} I spectral curve, which is a degenerate torus, and from topological recursion they get a PDE that annihilates the wave function, which is compatible with the isomonodromic system and, together with another identity coming from integrable systems, provides a quantum curve that annihilates the wave function. In \cite{Iwaki}, the first author slightly generalizes the same results to the case of any elliptic curve, that is he considers not only the degenerate case of Painlev\'{e} I, but tori where none of the two cycles are pinched. In both papers, they show that the $\hbar$ corrections from \cite{Bouchard_Quantizing} can be controlled by a derivative with respect to a deformation parameter. The quantum curves still contain infinitely many $\hbar$-correction terms, but in this case, these corrections are given by the asymptotic expansion of the solution of Painlev\'{e} I around $\hbar\rightarrow 0$.

In \cite{IwakiMarchalSaenz}, the approach is reversed: they prove that Lax pairs associated with $\hbar$-dependent Painlev\'{e} equations satisfy the topological type property of \cite{BergereBorotEynard}, which implies that one can reconstruct the $\hbar$-expansion of the isomonodromic $\tau$-function from topological recursion. Finally, in \cite{MarchalOrantin}, they generalize this result showing that it is always possible to deform a differential equation $\partial_x\Psi(x)=\mathcal L(x)\Psi(x)$, with $\mathcal L(x)\in\mathfrak{sl}_2(\mathbb{C})$ by introducing a formal parameter $\hbar$ in such a way that it satisfies the topological type property.

In the present work, we recover the PDE from \cite{IwakiSaenz,Iwaki} from loop equations (which are necessary for topological recursion, but not sufficient\footnote{Topological recursion provides a specific solution of loop equations. General solutions of loop equations don't necessarily satisfy the so-called projection formula, which implies that the purely holomorphic part of $\omega_{g,n}$ vanishes (in the sense of \cite[Section 2.2.2]{BSblob}), and they are governed by a generalisation called blobbed topological recursion, which was introduced in \cite{BSblob}.}) and as part of a system that we obtain because we consider a wave function where the integrals are over any divisor of degree zero. With our system, we are able to recover the whole isomonodromic system associated to Painlev\'{e} I just from loop equations. We also give an additional meaning to the deformation parameter that appears naturally in \cite{IwakiSaenz,Iwaki} for the case of elliptic curves, making use of the powerful idea of deforming with respect to the generalized cycles introduced in the previous section. The elliptic curve case is a very concrete case in which there is only one such deformation parameter, but we see that we need to consider several in the higher genus cases.

\subsubsection{Outline}

In Section~\ref{section2} we introduce the type of curves we consider in this work and relate them to the concept of spectral curves as input of the topological recursion. We also give the link to the spectral curves in the setting of isomonodromy systems, which serves as a motivation to us. Moreover, we compute the deformation parameters of the family of elliptic curves, which recovers the Painlev\'e I isomonodromy system setting in the degenerate case; in particular, the so-called KP times.

In Section~\ref{section3} we recall the loop equations for our specific setting and deduce some interesting consequences relating them to time deformations, which appear when considering spectral curves of genus $\hat{g}>0$.

In Section~\ref{section4} we prove our main result. From loop equations we obtain a system of partial differential equations that annihilates our wave function defined from topological recursion integrating over a general divisor. We also give the shape of this system in the particular cases of genus zero and elliptic curves.
Finally, we consider our system for the particular case of a two-point divisor, which then consists of only two PDEs, with differentials with respect to two spectral variables and the deformation parameters, called times. We are able to combine the two PDEs in such a way that we eliminate one of the spectral variables.

In Section~\ref{section5}, we argue that if the spectral curve comes from an isomonodromic system, then the topological recursion non-perturbative wave function has to coincide with the solution of the isomonodromic system, which implies an ODE, which is the quantum curve we were looking for. As particular interesting cases, we recover the first Painlev\'e system and equation, and its higher analogues defined in terms of Gelfand--Dikii polynomials.

In Appendix~\ref{appendix}, we prove that the integrable kernel of any isomonodromic system satisfies the same PDE that we found for the two-point wave function constructed from topological recursion.

\subsection*{Acknowledgements}This paper is partly a result of the ERC-SyG project, Recursive and Exact New Quantum Theory (ReNewQuantum) which received funding from the European Research Council (ERC) under the European Union's Horizon 2020 research and innovation programme under grant agreement No 810573. B.E.'s is also partly supported by the ANR grant Quantact: ANR-16-CE40-0017.
E.G.-F.~was supported by the public grant ``Jacques Hadamard'' as part of the
Investissement d'avenir project, reference ANR-11-LABX-0056-LMH,
LabEx LMH.

While we were finishing this manuscript, we learnt that O.~Marchal and N.~Orantin were working on the same topic using different methods, and we are grateful to them for informing us, and for valuable discussions about this. We plan to work together on the upcoming generalization to plane curves of arbitrary degrees\footnote{While the revised version of this article was produced, the first affirmative answer to the quantum curve conjecture for generic plane curves of arbitrary degrees was released \cite{EGMO21}.}.

We would also like to thank R.~Belliard, M.~Berg\`ere, G.~Borot, V.~Bouchard, J.~Zinn-Justin and A.~Voros for valuable discussions on the subject.

\section{Spectral curves with a global involution}\label{section2} 

In this article we focus on algebraic plane curves of the form
\beq
y^2=P(x),
\eeq
with $P(x)\in \mathbb C[x]$ an arbitrary polynomial of $x$, and we will generalize to $y^2=R(x)$ with $R(x)\in \mathbb C(x)$ an arbitrary rational function of $x$.

The degree of the polynomial is related to the genus of the curve. For example, in the case in which $P$ is a polynomial of degree $2m+1$ or $2m+2$ the curve has genus $\hat{g}\leq m$, with equality if the plane curve is smooth. 
If the degree is odd, the curve has one point at infinity and if the degree is even, the curve has two points at infinity. If $\hat{g}>1$, the curve is called hyperelliptic; if $\hat{g}=1$ (with a distinguished point), it is called elliptic, and if $\hat{g}=0$, it is called rational.

\subsection{Spectral curves as input of the topological recursion}
The method of topological recursion associates to a spectral curve $\mathcal{S}$ a doubly indexed family of meromorphic multi-differentials $\omega_{g,n}$ on $\Sigma^n$:
\begin{center}
TR:\ \ Spectral curve $\mathcal{S}=(\Sigma, x, y\dd x, B)$ $\leadsto$ Invariants $\omega_{g,n}$ ($F_g = \omega_{g,0}$).
\end{center}
A spectral curve is the data of a Riemann surface $\Sigma$, a holomorphic projection $x:\Sigma\to \mathbb C \mathrm{P}^1$ to the base $\mathbb C \mathrm{P}^1$ which turns $\Sigma$ into a ramified cover of the sphere, a meromorphic 1-form $y\dd x$ on $\Sigma$, and $B$ a 2nd kind fundamental differential, i.e.~a symmetric $1\boxtimes 1$ form on $\Sigma\times \Sigma$ with normalized double pole on the diagonal and no other pole, behaving near the diagonal as:
$$
B(z_1,z_2) = \frac{\dd z_1 \dd z_2}{(z_1-z_2)^2} + \text{holomorphic at }z_1=z_2.
$$
In case the spectral curve is of genus $0$, i.e.~$\Sigma = \mathbb{C}\mathrm{P}^1$, it is known that such a $B$ is unique and is worth
$$
B(z_1,z_2)=\frac{\dd z_1\dd z_2}{(z_1-z_2)^2}.
$$
If the genus $\hat{g}$ of $\Sigma$ is $\geq 1$, $B$ is not unique since one can add any symmetric bilinear tensor product of holomorphic 1-forms. 
A way to find a unique one is to choose a Torelli marking, which is a choice of a symplectic basis $\{\{\mathcal{A}_i\}_{i=1}^{\hat{g}},\{\mathcal{B}_i\}_{i=1}^{\hat{g}}\}$ of $H_1(\Sigma,\mathbb{Z})$. There exists a unique $B$  normalized on the $\mathcal{A}$-cycles of $H_1(\Sigma,\mathbb{Z})$ 
$$
\oint_{\mathcal{A}_i}B(z_1,\cdot)=0.
$$
Such a bi-differential has a natural construction in algebraic geometry and is called the normalized \emph{fundamental differential of the second kind} on $\Sigma$. 
See \cite{Eynard2018B} for constructing $B$ for general algebraic plane curves.

We define the \emph{filling fractions} as the $\mathcal{A}$-periods of $\omega_{0,1}$:
\beq\label{filling}
\epsilon_j\coloneqq \frac{1}{2\pi i}\oint_{\mathcal{A}_j}y \dd x, \text{ for } j=1,\ldots,\hat{g}.
\eeq

\begin{remark}
The coordinate $x: \Sigma \rightarrow \mathbb{C}\mathrm{P}^1$ in the definition of spectral curve can be thought as a ramified covering of the sphere. We call \emph{degree} of the spectral curve the number of sheets of the covering, i.e.~the number of preimages of a generic point. In this article, we focus on spectral curves of degree $2$ with a global involution $(x,y)\mapsto (x,-y)$.
\end{remark}

\subsection{Spectral curves from isomonodromic systems}

Painlev\'{e} transcendents have their origin in the study of special functions and of isomonodromic deformations of linear differential equations. They are solutions to certain nonlinear second-order ordinary differential equations in the complex plane with the Painlev\'{e} property, i.e.~the only movable singularities are poles.

An $\hbar$-dependent \emph{Lax pair} is a pair $(\mathcal L(x,t;\hbar),\mathcal R(x,t;\hbar))$ of $2\times 2$ matrices, whose entries are rational functions of $x$ and holomorphic in $t$ such that the system of partial differential equations
$$
\left\{
\begin{aligned}
\hbar \frac{\partial}{\partial x}\Psi(x,t) & = \mathcal L(x,t;\hbar)\Psi(x,t), \\
\hbar \frac{\partial}{\partial t}\Psi(x,t) & =\mathcal R(x,t;\hbar)\Psi(x,t) 
\end{aligned}
\right.
$$
is compatible. We call such a system an \emph{isomonodromy system}.

The compatibility condition, i.e.~$\frac{\partial}{\partial t\partial x}\Psi = \frac{\partial}{\partial x\partial t}\Psi$, is equivalent to the so-called zero-curvature equation:
$$
\hbar\frac{\partial \mathcal L}{\partial t}-\hbar\frac{\partial \mathcal R}{\partial x} + [\mathcal L,\mathcal R]=0.
$$
In \cite{JimboMiwa}, Jimbo and Miwa gave a list of the Lax pairs whose compatibility conditions are equivalent to the six Painlev\'{e} equations.

Let us consider the expansion around $\hbar =0$ of the first equation of the system: $\mathcal L(x,t;\hbar)= \sum_{k\geq 0}\hbar^k \mathcal L_k(x,t)$.
The associated \emph{spectral curve} is given by
\beq
\det(y\,\mathrm{Id}-\mathcal L_0(x,t))=0,
\eeq
which is actually a family of algebraic curves parametrized by $t$.

\subsubsection{Motivational example: The first Painlev\'{e} equation}

Let us consider the first Painlev\'{e} equation with a formal small parameter $\hbar$:
$$
P_I \colon \; \frac{\hbar^2}{2}\frac{\partial^2}{\partial t^2}U - 3U^2=t.
$$
The leading term $u=u(t)$ of a formal power series solution $U(t,\hbar) = u(t) + \sum_{k\geq 1}\hbar^{2k}u_k(t)$ satisfies $t=-3u^2$ and determines the subleading terms recursively:
$$
u_k =c_k u^{1-5k}, c_k \in\mathbb{Q}
$$
by the recursion
$$
c_0=1 \, , \quad\quad
2c_{k+1}  = \frac{25 k^2-1}{6^3} c_k - \sum_{j=1}^k c_j c_{k+1-j}.
$$
The coefficient $u_k(t)$ has a singularity at $u=0$, i.e.~$t=0$. This special point is called a turning point of $P_I$. We shall assume that $t\neq 0$.
We denote by $\dot{U}$ the derivative with respect to $t$ of $U(t,\hbar)$.

The Tau-function $\mathcal{T}(t)$ is defined in such a way that
$$
U(t) = -\hbar^2 \ \frac{\partial^2}{\partial t^2} \log\mathcal{T}.
$$
The Painlev\'e equation ensures that $\mathcal{T}$ is an entire function with simple zeros at the movable poles of $U$.

The Lax pair associated to the first Painlev\'{e} equation is given by
\beq\label{eqPainleveLax}
\mathcal L(x,t;\hbar)  \coloneqq \begin{pmatrix}
\frac{\hbar}{2}\dot{U} & x-U \\
(x-U)(x+2U)+\frac{\hbar^2}{2}\ddot{U} & -\frac{\hbar}{2}\dot{U}
\end{pmatrix}\text{ and }\,
\mathcal R(x,t;\hbar)  \coloneqq \begin{pmatrix}
0 & 1 \\
x+2U & 0
\end{pmatrix}.
\eeq

The leading term of $\mathcal L$ in its expansion around $\hbar =0$ is given by
$$
\mathcal L_0(x,t) = \begin{pmatrix}
0 & x-u \\
x^2+ux-2u^2 & 0
\end{pmatrix}.
$$
The spectral curve reads
\beq\label{degenerate}
\det(y\,\mathrm{Id}-\mathcal L_0(x,t))= y^2-(x-u)^2(x+2u)=0,
\eeq
which is actually a family of algebraic curves parametrized by $t$. Since we have assumed $t\neq 0$, the two roots $x=u$ and $x=-2u$ of $R(x)=(x-u)^2(x+2u)$ are distinct, but the root $x=u$ has multiplicity $2$. These curves have genus $0$ or, more precisely, constitute a family of tori with one of the cycles pinched.

In general, we want to study the family of tori given by
\beq
y^2=x^3+tx+V,
\eeq
where $R(x)=x^3+tx+V$ has three different roots.
The case $t=-3u^2$, $V=2u^3$ recovers the particular degenerate case \eqref{degenerate}.

\subsection{Parametrizations and deformation parameters of the elliptic case}
In the elliptic case, we give the parametrizations and compute the coordinates because it has more structure than the genus $0$ case, since we need to introduce one deformation parameter, and it also illustrates how the general case works. The degenerate case corresponds to Painlev\'{e} I, which is our prototypical example when making the connection to isomonodromy systems.
\subsubsection{Degenerate case}

When $V=2u^3$,
 consider the parametrization of the spectral curve given by 
$$
\begin{cases}
x(z)  = z^2-2u, \\
y(z)  = z^3-3uz,
\end{cases}
$$
with $z\in\Sigma=\mathbb{C}\mathrm{P}^1$ and the fundamental form of the second kind on $\Sigma$: 
$$B(z_1,z_2)=\frac{\dd z_1 \dd z_2}{(z_1-z_2)^2}.$$
It satisfies
$$
y^2 = x^3 + tx+V,
$$
with $t=-3u^2$, $V=2u^3$.

This curve has one ramification point at $z=0$ and one pole at $z=\infty$.

Near $z=\infty$, we have
$$
y \sim x^{\frac32} +\frac{t}{2x^{\frac12}}+\frac{V}{2x^{\frac32}} - \frac{t^2}{8 x^{\frac52}} - \frac{tV}{4 x^{\frac72}} + O(x^{-\frac92}).
$$
This implies that
$$
t_{\infty,1}=\frac{1}{2\pi i} \int_{\mathcal A_{\infty,1}} y\dd x = \underset{z\rightarrow\infty}{{\mathrm Res}} \ x^{-\frac12} \ y \dd x = - t,
$$
$$
\int_{\mathcal B_{\infty,1}} y \dd x = \underset{z\rightarrow\infty}{{\mathrm Res}} \ x^{\frac12} \ y \dd x = -V,
$$
$$
t_{\infty,5}=\frac{1}{2\pi i} \int_{\mathcal A_{\infty,5}} y\dd x = \underset{z\rightarrow\infty}{{\mathrm Res}} \ x^{-\frac52} \ y \dd x = -2,
$$
$$
\int_{\mathcal B_{\infty,5}} y \dd x = \frac{1}{5}\underset{z\rightarrow\infty}{{\mathrm Res}} \ x^{\frac52} \ y \dd x = \frac{tV}{10},
$$
where we use the generalized cycles $\mathcal A_{p,k}$ and $\mathcal B_{p,k}$ defined in \eqref{ABcycles} and we have considered only the values $k=1,5$ for which $t_{\infty,k}\neq 0$.

The degenerate cycle $\mathcal A$ corresponds to a small simple closed curve encircling $z_0=\sqrt{3u}$ and the cycle $\mathcal B$ corresponds to the chain $(-\sqrt{3u}\to\sqrt{3u})$. Therefore, we have
$$
\epsilon = \frac{1}{2\pi i} \oint_{\mathcal A} y \dd x = 0,
\quad \quad
I = \oint_{\mathcal B} y \dd x = \frac{-8}{15} \ (3u)^{\frac52}.
$$

The prepotential $F_0=\omega_{0,0}$, defined in general in \cite{GIS, EO07}, is worth
$$
F_0 = 
\frac12\bigg( \epsilon I  + \sum_k t_{\infty,k} \int_{\mathcal B_{\infty,k}} y\dd x \bigg)
=\frac12\bigg(tV - 2 \frac{tV}{10}\bigg)
=\frac{2}{5} tV = -\frac{12}{5}u^5
$$
and satisfies
$$
\frac{\partial F_0}{\partial t} = 2u^3=V,
\quad  \quad
\frac{\partial^2 F_0}{\partial t^2} = \frac{\partial V}{\partial t}=-u.
$$

\subsubsection{Non-degenerate case: elliptic curves}

When $V\neq 2u^3$, we shall consider the Weierstrass parametrization of the torus of modulus $\tau$, and with a scaling $\nu$:
$$
\begin{cases}
x(z)  = \nu^2 \wp(z), \\
y(z)  = \frac{\nu^3}{2} \wp'(z),
\end{cases}
$$
with $z\in\Sigma=\mathbb{C}/\mathbb{Z}+\tau\mathbb{Z}$ the torus of modulus $\tau$.
The fundamental form of the second kind on $\Sigma$, normalized on the $\mathcal A$ cycle is: 
$$
B(z_1,z_2)=(\wp(z_1-z_2)+G_2(\tau))\dd z_1 \dd z_2.$$ 
with $G_k(\tau)$ the $k^{\rm th}$ Eisenstein series.
It satisfies
$$
y^2 = x^3 + tx+V,
$$
with 
$$
t = -15 \nu^4 G_4(\tau),
\quad \quad
V= -35 \nu^6  G_6(\tau).
$$
Instead of parametrizing the spectral curve with $t$ and $V$, we shall parametrize it with $t$ and $\epsilon$, where
$$
\epsilon = \frac{1}{2\pi i} \oint_{\mathcal A} y\dd x = \frac{3 \nu^5}{\pi i}\big(2G_2(\tau)G_4(\tau)-7G_6(\tau)\big)=- 3 \nu^5 G'_4(\tau).
$$
We shall now write
$$
V=V(t,\epsilon).
$$
We have
$$
\dd V = \frac53 \pi i \nu \dd\epsilon - \frac23 \nu^2 G_2(\tau) \dd t.
$$

This curve has 3 ramification points at $z=\frac12, \frac\tau{2},\frac{1+\tau}{2}$, and one pole at $z=0$.

Near $z=0$, we have
$$
y \sim x^{\frac32} +\frac{t}{2x^{\frac12}}+\frac{V}{2x^{\frac32}} - \frac{t^2}{8 x^{\frac52}} - \frac{tV}{4 x^{\frac72}} + O(x^{-\frac92}).
$$
This implies that
$$
t_{\infty,1}=\frac{1}{2\pi i} \int_{\mathcal A_{\infty,1}} y\dd x = \underset{z\rightarrow 0}{{\mathrm Res}} \ x^{-\frac12} \ y\dd x = -t,
$$
$$
\int_{\mathcal B_{\infty,1}} y\dd x = \underset{z\rightarrow 0}{{\mathrm Res}} \ x^{\frac12} \ y\dd x = -V.
$$
$$
t_{\infty,5}=\frac{1}{2\pi i} \int_{\mathcal A_{\infty,5}} y\dd x = \underset{z\rightarrow 0}{{\mathrm Res}} \ x^{-\frac52} \ y \dd x = -2,
$$
$$
\int_{\mathcal B_{\infty,5}} y \dd x = \frac{1}{5}\underset{z\rightarrow 0}{{\mathrm Res}} \ x^{\frac52} \ y \dd x = \frac{tV}{10}.
$$
We also consider
$$
I = \oint_{\mathcal B} y\dd x .
$$
The prepotential $F_0=\omega_{0,0}$ is worth
$$
F_0 = \frac12\left( t V -2\frac{tV}{10}+I\epsilon\right) =  \frac25 tV +\frac12 I\epsilon,
$$
and satisfies
$$
\dd F_0 = V \dd t+I \dd \epsilon .
$$

\begin{remark}
In particular, for elliptic curves, we can express the independent term of the curve $V$ as a deformation of the prepotential $\omega_{0,0}$ with respect to the coefficient of the linear term $t$:
\beq \label{prepotentialV}
\frac{\partial}{\partial t} \omega_{0,0}= V.
\eeq
This is a consequence of $\frac{\partial}{\partial t_{\infty,1}}\omega_{0,0}=\int_{\mathcal{B}_{\infty,1}}\omega_{0,1}$ (which is already proved in \cite[Theorem~5.1]{EO07}).
\end{remark}

In terms of the torus modulus $\tau$ and scaling $\nu$, we have
$$
t=-15 \nu^4 G_4(\tau),
\quad  \quad
V= -35 \nu^6 G_6(\tau),
\quad  \quad
\epsilon =-3 \nu^5 G'_4(\tau),
\quad  \quad
I= 2\pi i \tau \epsilon + \frac45 \nu t.
$$

We also have
$$
\omega_{1,0}=F_1 = \frac{1}{48}\log{(4t^3+27 V^2)} + \frac{1}{4} \log{\frac{2}{\nu}}.
$$

\section{Loop equations and deformation parameters}\label{section3}

We start by recalling the loop equations for the topological recursion applied to any spectral curve of degree $2$ with a global involution. Let $y^2=R(x)$, with $R\in\mathbb{C}(x)$. The family of curves that we consider has the global involution $z\mapsto -z$, i.e.~$x(z)=x(-z)$.

Let $\omega_{0,1}(z) \coloneqq y(z) \dd x(z)$, $\omega_{0,2}(z_1,z_2)\coloneqq B(z_1,z_2)$ and $\omega_{g,n}$ for $2g-2+n>0$ be defined as the topological recursion amplitudes for this initial data \cite{EO07}.

The loop equations for this particular case read:

\begin{theorem}\cite[Proposition~2.8]{BEO}\label{loop}
Let $g, n\in \mathbb{N}$. The linear loop equations read:
\beq 
\omega_{g,n+1}(z,z_1,\ldots,z_n)+\omega_{g,n+1}(-z,z_1,\ldots,z_n)
= \delta_{g,0}\delta_{n,1} \ \frac{\dd x(z)\dd x(z_1)}{(x(z)-x(z_1))^2}.
\eeq

The quadratic loop equations claim that the following expression
\beq\label{quadraticloop}
\frac{1}{\dd x(z)^2}\Bigg(\omega_{g-1,n+2}(z,-z,z_1,\ldots,z_n)+\sum_{\substack{g_1+g_2=g, \\ I_1\sqcup I_2 = \{z_1,\ldots,z_n\}}} \omega_{g_1,1+|I_1|}(z,{I_1})\omega_{g_2,1+|I_2|}(-z,{I_2})\Bigg)
\eeq
is a rational function of $x(z)$ with no poles at the branch-points.
\end{theorem}
We will make use of an immediate consequence of the loop equations:
\begin{corollary}
For all $g, n \geq 0$,
\bea\label{loopeq}
P_{g,n}(x(z);z_1,\ldots,z_n) & \coloneqq &\frac{-1}{\dd x(z)^2}\Bigg(\omega_{g-1,n+2}(z,-z,z_1,\ldots,z_n) \\ 
 & + &\sum_{\substack{g_1+g_2=g, \\ I_1\sqcup I_2 = \{z_1,\ldots,z_n\}}} \omega_{g_1,1+|I_1|}(z,{I_1})\omega_{g_2,1+|I_2|}(-z,{I_2})\Bigg) \nonumber\\ 
&  +  & \sum_{i=1}^n \dd_i \left( \frac{1}{x(z)-x(z_i)}\frac{\omega_{g,n}(z_1,\ldots,-z_i,\ldots,z_n)}{\dd x(z_i)}\right) \nonumber
\eea
is a rational function of $x(z)$ that
has no poles at the branch-points and no poles when $x(z)=x(z_i)$.
\end{corollary}
\begin{proof}
First, the expression \eqref{loopeq} is invariant under $z\to -z$ and meromorphic on $\Sigma$, hence a meromorphic function of $x(z)\in \mathbb C\mathrm P^1 $, i.e.~a rational function of $x(z)$.
From the loop equations, it has no pole at branchpoints.
Let us study the behavior at $z=z_i$.
The only term in~\eqref{quadraticloop} that contains a pole at $z=z_i$ is
$$
\frac{1}{\dd x(z)^2} B(z,z_i)\,\omega_{g,n}(-z,z_1,\ldots,\hat{z_i},\ldots,z_n).
$$
Remark that $B(-z,z_i)$ has no pole at $z=z_i$ and 
$$B(z,z_i)+B(-z,z_i) = \frac{\dd x(z) \dd x(z_i)}{(x(z)-x(z_i))^2} = \dd_i \left( \frac{\dd x(z)}{x(z)-x(z_i)}\right).
$$
Therefore we add a term without any poles to the previous one and consider the term with a pole at $z=z_i$ to be
$$
\frac{1}{\dd x(z)^2}\frac{\dd x(z) \dd x(z_i)}{(x(z)-x(z_i))^2}\,\omega_{g,n}(-z,z_1,\ldots,\hat{z_i},\ldots,z_n).
$$
We can write it as 
\begin{multline}
\dd_i\bigg(\left(\frac{1}{x(z)-x(z_i)}\right) \bigg(\frac{\omega_{g,n}(-z,z_1,\ldots,\hat{z_i},\ldots,z_n)}{\dd x(z)}\\ -\frac{\omega_{g,n}(-z_i,z_1,\ldots,\hat{z_i},\ldots,z_n)}{\dd x(z_i)}+\frac{\omega_{g,n}(-z_i,z_1,\ldots,\hat{z_i},\ldots,z_n)}{\dd x(z_i)}\bigg)\bigg).
\end{multline}
The sum of the first two terms does not have a pole at $z=z_i$. Therefore subtracting the last term for all $i=1,\ldots,n$, we obtain an expression with no poles at $z=z_i$.
Since this expression is an even function of $z$, there is no pole at $z=-z_i$ either.
\end{proof}

Recall that in general $\omega_{0,2}=B$ can have poles only at coinciding points and the $\omega_{g,n}$'s with $2g-2+n>0$ can have poles only at ramification points. Therefore, from the corollary we see that $P_{g,n}(x(z);z_1,\ldots,z_n)$ as a function of $z$ can only have poles at the poles of $\omega_{0,1}=y\dd x$.

\subsection{Variational formulas of the topological recursion} We recall one of the most important properties of the topological recursion: we know how the output data behaves when we vary the family of spectral curves to which we apply the topological recursion with respect to some parameters. This result already appeared in the original paper where topological recursion was introduced \cite[Theorem~5.1]{EO07} and was revisited in terms of generalised cycles in \cite[Theorem~3.2]{GIS}. A detailed proof of it with applications in the context of combinatorics can be found in \cite[Section~3.2]{BCCG22}.

\begin{theorem}\label{variational}
Let $(\Sigma,x_t,y_t\dd x_t,B_{\Sigma})$ be a holomorphic family of spectral curves, where $B_{\Sigma}$ is the normalised bidifferential of the second kind on $\Sigma$, depending on a parameter $t\in\mathbb{C}$. Let $\omega_{g,n}^t$ be the topological recursion differentials associated to that family of spectral curves. For the base topologies, we have $\omega_{0,1}^t=y_t\dd x_t$ and $\omega_{0,2}^t=B_{\Sigma}$. Assume there exists a generalised cycle $\gamma$ on $\Sigma$ whose support does not contain the zeroes of $\dd x$ and such that
$$
\partial_t\, \omega_{0,1}^t(z)=\int_{\gamma}B(z,\cdot),
$$
where the derivatives are taken with $x$ fixed. Then,
\begin{equation}\label{variational-eq}
\partial_t\, \omega_{g,n}^t(z_1,\ldots,z_n)=\int_{\gamma}\omega_{g,n+1}^t(z_1,\ldots,z_n,\cdot),
\end{equation}
where the derivatives are taken at $x_i=x(z_i)$ fixed.
\end{theorem}

We remark that the equation \eqref{variational-eq} for $(0,2)$ follows from the formula for $\omega_{0,3}$ in terms of $B_{\Sigma}$ \cite[Theorem~4.1]{EO07} and the Rauch variational formula (see \cite[Lemma~3.13]{BCCG22} and references therein).


\subsection{Relation to time deformation for elliptic curves} Now we restrict ourselves to curves described by polynomials of the form
$$
y^2=x^3+tx+V.
$$ 
In this case, $\omega_{0,1}=y\dd x$ can only have poles at $x(z)=\infty$.

The topological recursion amplitudes for $2g-2+n\geq 0$ are analytic away from branchpoints; in particular they are analytic at $\infty$, with the following behavior near $z=\infty$:
$$
\omega_{g,n}(z,z_1,\ldots,z_n)=O(z^{-2}).
$$

In our case, this implies the following behavior at $x(z)=\infty$: 
\beq 
\frac{\omega_{g,n+1}(z,z_1,\ldots,z_n)}{\dd x(z)} = O(x(z)^{\frac{-3}{2}}), \text{ for } 2g-2+n\geq 0.
\eeq

Since the only pole can come from terms that contain $\omega_{0,1}=y\dd x$, $P_{g,n}$ has the following behavior at $x(z)\rightarrow\infty$:
$$
P_{g,n}(x(z);z_1,\ldots,z_n)= 2\frac{y(z) \dd x(z)\, \omega_{g,n+1}(z,z_1,\ldots,z_n)}{\dd x(z)^2} + O(x(z)^{-1}),
$$
i.e.
\begin{equation}
P_{g,n}(x(z);z_1,\ldots,z_n)=2y(z) O(x(z)^{\frac{-3}{2}}) + O(x(z)^{-1})= O(x(z)^0), \text{ for } 2g-2+n\geq 0,
\end{equation}
where the last behavior comes from the fact that in the elliptic curve case, we have $y\sim x^{\frac32}$.

We have seen that $P_{g,n}$ is a polynomial of degree $0$, that is independent of $z$, and can be written:
$$
P_{g,n}(x(z);z_1,\ldots,z_n)= 2\lim_{z\to\infty} x(z)^{\frac32} \  \frac{\omega_{g,n+1}(z,z_1,\ldots,z_n)}{\dd x(z)}.
$$

\begin{corollary}
For $(g,n)\neq (0,0),(0,1)$:
\beq\label{P}
P_{g,n}(x(z);z_1,\ldots,z_n) = -\oint_{\mathcal{B}_{\infty,1}}\omega_{g,n+1}(z,z_1,\ldots,z_n)=\frac{\partial}{\partial t}\omega_{g,n}(z_1,\ldots,z_n),
\eeq
with $\mathcal{B}_{\infty,1}$ the second kind cycle given by $\frac{1}{2\pi i}\mathcal{C}_{\infty} \sqrt{x(z)}$, where $\mathcal{C}_{\infty}$ denotes a small contour around~$\infty$.

Moreover
\beq\label{P0}
P_{0,0}(x(z)) = y(z)^2 = x^3+tx+V = x^3+tx+\frac{\partial}{\partial t} \omega_{0,0},
\eeq
\beq\label{P1}
P_{0,1}(x(z);z_1)= 2\frac{y(z)}{\dd x(z)}B(z,z_1) - \dd_1\left(\frac{y(z)+y(z_1)}{x(z)-x(z_1)} \right)= \frac{\partial}{\partial t} \omega_{0,1}(z_1).
\eeq
\end{corollary}

\begin{proof} For $(g,n)=(0,0)$, we obtain $P_{0,0}(x(z)) = y(z)^2$ from $\eqref{loopeq}$, which together with $V=\frac{\partial}{\partial t}\omega_{0,0}$ from \eqref{prepotentialV} gives the expression for this case.

The expression for $(g,n)=(0,1)$ is a direct computation using \eqref{loopeq}:
$$
P_{0,1}(x(z);z_1)=\frac{y(z)}{\dd x(z)}(B(z,z_1)-B(-z,z_1))-\dd_1\left(\frac{y(z_1)}{x(z)-x(z_1)}\right).
$$
In order to get the second equality in \eqref{P1}, observe from the first equality that also $P_{0,1}$ has the following behavior at $x(z)\rightarrow\infty$:
$$
P_{0,1}(x(z);z_1)= 2\frac{y(z) \, B(z,z_1)}{\dd x(z)} + O(x(z)^{-1})=2y(z) O(x(z)^{\frac{-3}{2}}) + O(x(z)^{-1})= O(x(z)^0).
$$

For $(g,n)\neq (0,0)$, since $P_{g,n}$ is constant with respect to $z$, we can write
\bea
P(x(z);z_1,\ldots,z_n) &=& 2 \lim_{x(z)\to\infty} x(z)^{\frac32} \ \frac{\omega_{g,n+1}(z,z_1,\ldots,z_n)}{\dd x(z)} \nonumber \\
&=& -\underset{z\rightarrow 0}{{\mathrm Res}}\;\; \sqrt{x(z)}\; \omega_{g,n+1}(z,z_1,\ldots,z_n)=\nonumber \\ 
&=&-\oint_{\mathcal{B}_{\infty,1}}\omega_{g,n+1}(z,z_1,\ldots,z_n).
\eea
Moreover, since $t=-t_{\infty,1}$, we have $\frac{\partial}{\partial t}\omega_{0,1}=-\frac{\partial}{\partial t_{\infty,1}}\omega_{0,1}=-
\int_{\mathcal B_{\infty,1}} \omega_{0,2}$, and since $B=\omega_{0,2}$ is the Bergman kernel normalized on the $\mathcal A$-cycles, 
Theorem~\ref{variational} implies for any topology that 
\beq
\oint_{\mathcal{B}_{\infty,1}}\omega_{g,n+1}(z,z_1,\ldots,z_n)=\frac{\partial}{\partial t_{\infty,1}}\omega_{g,n}(z_1,\ldots,z_n)=-\frac{\partial}{\partial t}\omega_{g,n}(z_1,\ldots,z_n).\nonumber
\eeq
\end{proof}

\subsection{Generalization to any plane curve with a global involution} Our goal is to find the relation between $P_{g,n}$ from loop equations and an operator depending on time deformations acting on the topological recursion amplitudes. In this section, we generalize the relation that we have just found in the Painlev\'{e}~I case to all plane curves with a global involution.
 
Consider algebraic curves of the form
\beq
y^2=R(x),
\eeq
with $R(x)\in \mathbb C(x)$ an arbitrary rational function. 

This is parametrized by a pair of meromorphic functions $x,y$ on a Riemann surface $\Sigma$.
The ramified covering given by $x: \Sigma \to \mathbb{C}\mathrm{P}^1$ is a double cover.
Depending on the parity of the behavior of $y$ at $x\to \infty$, $x$ has either one double pole (order $d=-2$) or 2 simple poles (order $d=-1$).
Let us denote $\sigma$ the global involution which sends $(x,y)\mapsto (x,-y)$.

Let $\{\lambda_l\}_{l=1}^{N}$ be the set of zeroes of the denominator of $R(x)$. The $1$-form $\omega_{0,1}=y\dd x$ can have a pole over $x=\infty$ and poles over $\{\lambda_l\}_{l=1}^{N}$. We call $\zeta_i$ the poles of $\omega_{0,1}$ of respective degrees given by $m_i$.

Let us define $d_i\coloneqq \mathrm{ord}_{\zeta_i}(x)$. If $\zeta_i$ is not a pole of $x$, we assume that it is not a ramification point, hence  $d_i=1$. If $\zeta_i$ is a pole of $x$, then $d_i$ can be either $-2$ or $-1$ as we commented, depending on $\zeta_i$ being a ramification point or not. Notice that if $\zeta_i$ is a pole, so is $\sigma(\zeta_i)$, and thus poles come in pairs. Moreover, we can only have $\sigma(\zeta_i)=\zeta_i$, if $\zeta_i$ is a double pole of $x$, which corresponds to the case $d_i=-2$.
Near $\zeta_i$ we use the local variable
\beq
\xi_i = (x-x(\zeta_i))^{\frac{1}{d_i}},
\eeq
where we define $x(\zeta_i)=0$, if $\zeta_i$ is a pole of $x$.

We write the Laurent expansion:
\beq\label{KPexpansionydx}
y\dd x \sim \sum_{j=0}^{m_i} t_{\zeta_i,j} \ \xi_i^{-1-j} \dd\xi_i + \text{analytic at } \zeta_i.
\eeq
This defines the local KP times \cite{GIS}
\beq
t_{\zeta_i,j} = \Res_{\zeta_i} \xi_i^j y\dd x
= \frac{1}{2\pi i} \oint_{{\mathcal A}_{\zeta_i,j}} y\dd x.
\eeq
Notice that for poles for which $\sigma(\zeta_i)\neq \zeta_i$, we have
\beq
t_{\sigma(\zeta_i),j}=-t_{\zeta_i,j}.
\eeq

Again, loop equations imply that the $P_{g,n}(x;z_1,\dots,z_n)$ defined in \eqref{loopeq} has no pole at coinciding points or at branchpoints.
It must be a rational function of $x$, whose poles can be only at the poles of $y\dd x$, i.e.~at the poles of $R(x)$ and possibly at $x=\infty$.

For every second type cycle $\mathcal{B}_{p,k}$, $k\geq 1$, we define the operator $\d_{\mathcal{B}_{p,k}}$ as
$$
\d_{\mathcal{B}_{p,k}}\omega_{g,n}(z_1,\ldots,z_n)\coloneqq \int_{\mathcal{B}_{p,k}}\omega_{g,n+1}(\cdot,z_1,\ldots,z_n) =\Res_{x\to x(p)} \frac{\xi_p^{-k}}{k}\omega_{g,n+1}(\cdot,z_1,\ldots,z_n).
$$
\begin{proposition}
Let $\zeta_i$ be the poles of $\omega_{0,1}$ of respective degrees given by $m_i$. The operator
\bea
L(x) &\coloneqq &
 \sum_{i, x(\zeta_i)=\infty}  \sum_{j=1-2d_i}^{m_i} t_{\zeta_i,j} \sum_{0\leq k \leq \frac{1-j}{d_i}-2} x^k  \Big(-\frac{j}{d_i}-k-2\Big) \partial_{\mathcal{B}_{\zeta_i,j+d_i(k+2)}}   \label{operatorL}\\
&& +  \sum_{i, x(\zeta_i)\neq \infty} \sum_{j=0}^{m_i} t_{\zeta_i,j} \sum_{k=0}^{j} (x-x(\zeta_i))^{-(k+1)} (j+1-k) \partial_{\mathcal{B}_{\zeta_i,j+1-k}} \nonumber
\eea
gives
\beq\label{operatorLeq}
P_{g,n}(x;z_1,\dots,z_n) = L(x).\omega_{g,n}(z_1,\dots,z_n).
\vspace{0.1cm}
\eeq

\end{proposition}
\begin{proof}
Let us write the Cauchy formula for $x=x(z)$
\bea
P_{g,n}(x;z_1,\dots,z_n)
&=& \Res_{x'\to x} \frac{\dd x'}{x'-x} \ P_{g,n}(x';z_1,\dots,z_n) \nonumber \\
&=& \frac12 \Res_{z'\to z} \frac{\dd x(z')}{x(z')-x(z)} \ P_{g,n}(x(z');z_1,\dots,z_n) \nonumber \\
&& + \frac12 \Res_{z'\to \sigma(z)} \frac{\dd x(z')}{x(z')-x(z)} \ P_{g,n}(x(z');z_1,\dots,z_n) \nonumber \\
&=& \frac12\sum_i \Res_{z'\to \zeta_i} \frac{\dd x(z')}{x(z)-x(z')} \ P_{g,n}(x(z');z_1,\dots,z_n) \nonumber \\
&=& -\frac12 \sum_{i, x(\zeta_i)=\infty} \sum_{k\geq 0} x(z)^k \Res_{z'\to \zeta_i} x(z')^{-(k+1)}  \dd x(z')  \ P_{g,n}(x';z_1,\dots,z_n) \nonumber \\
&&+ \frac12 \sum_{i, x(\zeta_i)\neq \infty} \sum_{k\geq 0} \xi_i(z)^{-(k+1)} \Res_{z'\to \zeta_i} \xi_i(z')^{k}\dd x(z')    \ P_{g,n}(x';z_1,\dots,z_n) .
\eea
Since the behavior at the poles is given by the terms containing $\omega_{0,1}=y\dd x$ in \eqref{loopeq},  near any of the poles $\zeta_i$  we have
\beq
P_{g,n}(x(z);z_1,\dots,z_n)
\sim \frac{2y(z)}{\dd x(z)} \omega_{g,n+1}(z,z_1,\dots,z_n) + O(\xi_i^{-2(d_i+1)}) .
\eeq
First consider the poles over $x=\infty$, with $d_i=-1$ or $d_i=-2$, which contribute to $P_{g,n}$ as 
\bea
&& - \sum_{i, x(\zeta_i)=\infty} \sum_{k\geq 0} x(z)^k \Res_{z'\to \zeta_i} x(z')^{-(k+1)}  y(z')  \omega_{g,n+1}(z',z_1,\dots,z_n) \nonumber \\
&=& - \sum_{i, x(\zeta_i)=\infty} \sum_{j=0}^{m_i} t_{\zeta_i,j} \sum_{k\geq 0} x(z)^k \Res_{z'\to \zeta_i} \xi_i(z')^{-d_i(k+1)}  \xi_i(z')^{-j-1} \frac{1}{d_i\, \xi_i(z')^{d_i-1}} \omega_{g,n+1}(z',z_1,\dots,z_n) \nonumber \\
&=& - \sum_{i, x(\zeta_i)=\infty} \frac{1}{d_i} \sum_{j=0}^{m_i} t_{\zeta_i,j} \sum_{k\geq 0} x(z)^k \Res_{z'\to \zeta_i} \xi_i(z')^{-d_i(k+2)-j}  \omega_{g,n+1}(z',z_1,\dots,z_n) \nonumber \\\
&=&   -\sum_{i, x(\zeta_i)=\infty}  \sum_{j=0}^{m_i} t_{\zeta_i,j} \sum_{k\geq 0} x(z)^k  \Big(k+2+\frac{j}{d_i}\Big) \int_{\mathcal B_{\zeta_i,j+d_i(k+2)}}  \omega_{g,n+1}(z',z_1,\dots,z_n) \nonumber \\
&=& - \sum_{i, x(\zeta_i)=\infty}  \sum_{j=1-2d_i}^{m_i} t_{\zeta_i,j} \sum_{0\leq k\leq \frac{1-j}{d_i}-2} x(z)^k  \Big(k+2+\frac{j}{d_i}\Big)\partial_{\mathcal{B}_{\zeta_i,j+d_i(k+2)}}  \omega_{g,n}(z_1,\dots,z_n) .\nonumber
\eea
The finite poles contribute as
\bea
&& \frac12 \sum_{i, x(\zeta_i)\neq \infty} \sum_{k\geq 0} \xi_i(z)^{-(k+1)} \Res_{z'\to \zeta_i} \xi_i(z')^{k}dx(z')    \ P_{g,n}(x';z_1,\dots,z_n) \nonumber \\
&=&  \sum_{i, x(\zeta_i)\neq \infty} \sum_{k\geq 0} \xi_i(z)^{-(k+1)} \Res_{z'\to \zeta_i} \xi_i(z')^{k}   y(z')  \omega_{g,n+1}(z',z_1,\dots,z_n) \nonumber \\
&=&  \sum_{i, x(\zeta_i)\neq \infty} \sum_{j=0}^{m_i} t_{\zeta_i,j} \sum_{k\geq 0} \xi_i(z)^{-(k+1)} \Res_{z'\to \zeta_i} \xi_i(z')^{k-j-1}    \omega_{g,n+1}(z',z_1,\dots,z_n) \nonumber \\
&=&  \sum_{i, x(\zeta_i)\neq \infty} \sum_{j=0}^{m_i} t_{\zeta_i,j} \sum_{k=0}^{j} \xi_i(z)^{-(k+1)} (j+1-k) \int_{\mathcal B_{\zeta_i,j+1-k}}     \omega_{g,n+1}(z',z_1,\dots,z_n) \nonumber \\
&=&  \sum_{i, x(\zeta_i)\neq \infty} \sum_{j=0}^{m_i} t_{\zeta_i,j} \sum_{k=0}^{j} \xi_i(z)^{-(k+1)} (j+1-k) \partial_{\mathcal{B}_{\zeta_i,j+1-k}}    \omega_{g,n}(z_1,\dots,z_n) . \nonumber
\eea
\end{proof}

Let us now rewrite the operator $L(x)$ as a differential operator with respect to the moduli of the spectral curve, i.e.~the times appearing in the pole structure of $\omega_{0,1}$ and the poles $\Lambda=\{\lambda_l\}_{l=1}^N$ of~$R(x)$.

\begin{lemma}
Let $\zeta_{l}^{(1)}$ and $\zeta_{l}^{(2)}$ be the two preimages of $\lambda_l$ by $x$, which are poles of order $m_l$ of $\omega_{0,1}$. We assume that the filling fractions $\epsilon_i$ are independent of the poles $\lambda_l$, i.e.
$$
\frac{\d}{\d\lambda_l}\epsilon_i=\oint_{\mathcal{A}_i}\frac{\d}{\d\lambda_l}\omega_{0,1}=0.
$$ 
Then we can decompose the operator derivative with respect to $\lambda_l$ as
\beq\label{lambdaOp}
\frac{\d}{\d\lambda_l}\omega_{g,n}(z_1,\ldots,z_n)=\sum_{i=1,2}\sum_{j=0}^{m_l}(j+1)t_{\zeta_l^{(i)},j}\partial_{\mathcal{B}_{\zeta_l^{(i)},j+1}}\omega_{g,n}(z_1,\ldots,z_n),
\eeq
for all $g, n \geq 0$.
\end{lemma}
\begin{proof}
Let $\zeta_{l}$ be one of the two preimages of $\lambda_l$ by $x$ and $\xi_l=x-\lambda_l$ be the local variable around it. Then, for $z\rightarrow\zeta_{l}$, we have
$$
\frac{\d}{\d\lambda_l}\omega_{0,1}(z)  \sim\sum_{j=0}^{m_l}(j+1)t_{\zeta_l,j} (\xi_l(z))^{-j-2}\dd \xi_l + \text{ analytic}.  \\
$$
On the other hand, we have
\begin{align*}
\partial_{\mathcal{B}_{\zeta_l,j}}\omega_{0,1}(z) & = \int_{\mathcal{B}_{\zeta_l,j}}\omega_{0,2}(\cdot,z) = \underset{z'=\lambda_l}{\mathrm{Res}}\, \frac{\omega_{0,2}(z',z)}{j(\xi_p(z'))^j} \\ 
&= -\underset{z'=z}{\mathrm{Res}}\, \frac{\omega_{0,2}(z',z)}{j(\xi_p(z'))^j} + \text{ analytic} = \frac{\dd \xi_l(z)}{(\xi_l(z))^{j+1}} + \text{ analytic}.
\end{align*}
The last equality follows from the fact that the fundamental differential of the second kind $\omega_{0,2}$ only has a double pole along the diagonal and the next-to-last equality follows from the decomposition
$$
\oint_{\gamma_{\lambda_{l}}} \frac{\omega_{0,2}(\cdot,z)}{j(\xi_p(\cdot))^j}=\oint_{\gamma_{z,\lambda_{l}}} \frac{\omega_{0,2}(\cdot,z)}{j(\xi_p(\cdot))^j}-\oint_{\gamma_z}\frac{\omega_{0,2}(\cdot,z)}{j(\xi_p(\cdot))^j}=\text{ analytic} -\oint_{\gamma_z}\frac{\omega_{0,2}(\cdot,z)}{j(\xi_p(\cdot))^j},
$$
where $\gamma_z$ and $\gamma_{\lambda_{l}}$ are small counterclocwise contours around $z$ and $\lambda_l$, and $\gamma_{z,\lambda_{l}}$ is a counterclocwise contour that surrounds the segment from $z$ to $\lambda_l$, which gives rise to an analytic term.

Therefore,
$$
\frac{\d}{\d\lambda_l}\omega_{0,1}(z)-\sum_{i=1,2}\sum_{j=0}^{m_l}(j+1)t_{\zeta_l^{(i)},j}\partial_{\mathcal{B}_{\zeta_l^{(i)},j+1}}\omega_{0,1}(z)
$$
is a holomorphic $1$-form with vanishing $\mathcal{A}$-periods and thus it vanishes. 

This yields the equality for $\omega_{0,1}$. Since we chose $\omega_{0,2}$ to be the Bergman kernel with vanishing $\mathcal{A}$-periods, we obtain the equality for the rest of the $\omega_{g,n}$ from Theorem~\ref{variational}.
\end{proof}

\begin{corollary}\label{isomonodromic} Let $\zeta_{\infty}\in x^{-1}(\infty)$ and $\zeta_{l}\in x^{-1}(\lambda_l)$ be poles of $\omega_{0,1}$ of orders $m_{\infty}$ and $m_l$, $l=1,\ldots,N$, respectively. Let $d_{\infty}\coloneqq \mathrm{ord}_{\zeta_{\infty}}(x)$. The operator \eqref{operatorL} is equal to the differential operator $L(x)=L_{\infty}(x)+L_{\Lambda}(x)$, with
\beq\label{inftyPart}
L_{\infty}(x)=  \sum_{j=1-2d_{\infty}}^{m_{\infty}} t_{\zeta_{\infty},j} \sum_{k= 0}^{\frac{1-j}{d_{\infty}}-2} x^k  \Big(-\frac{j}{d_{\infty}}-k-2\Big) \frac{\partial}{\partial t_{\zeta_{\infty},j+d_{\infty}(k+2)}}
\eeq
and
\beq\label{finitePart}
L_{\Lambda}(x) =  \sum_{l=1}^N \left(\frac{1}{x-\lambda_l}\frac{\d}{\d\lambda_l}+ \sum_{j=1}^{m_l-1} t_{\zeta_l,j} \sum_{k=1}^{j} (x-\lambda_l)^{-(k+1)} (j+1-k) \frac{\d}{\partial t_{\zeta_l,j+1-k}} \right).  
\eeq
\end{corollary}
\begin{proof}
From the variational formulas of topological recursion (see Theorem~\ref{variational}), if $\zeta$ is a ramified pole of $\omega_{0,1}$ of order $m$, i.e.~$x(\zeta)=\infty$ and $\mathrm{ord}_{\zeta}(x)=-2$, then we have
$$
\frac{\d}{\d t_{\zeta,j}}\omega_{g,n}(z_1,\ldots,z_n) = \d_{\mathcal{B}_{\zeta,j}}\omega_{g,n}(z_1,\ldots,z_n), \text{ for } j=1,\ldots,m.
$$
On the other hand, if $\zeta$ is an unramified pole of $\omega_{0,1}$ of order $m$, either $x^{-1}(\infty)=\{\zeta,\sigma(\zeta)\}$ and $\mathrm{ord}_{\zeta}(x)=-1$ or $x^{-1}(\lambda_l)=\{\zeta,\sigma(\zeta)\}$ and $\mathrm{ord}_{\zeta}(x)=1$, then the coefficients of the expansion of $\omega_{0,1}$ around $\zeta=\zeta_+$ and the other preimage $\sigma(\zeta)=\zeta_-$ are not independent. Therefore, the variational formulas include residues at the two preimages:
$$
\frac{\d}{\d t_{\zeta,j}}\omega_{g,n}(z_1,\ldots,z_n) = \d_{\mathcal{B}_{\zeta,j}}\omega_{g,n}(z_1,\ldots,z_n)- \d_{\mathcal{B}_{\sigma(\zeta),j}}\omega_{g,n}(z_1,\ldots,z_n), \text{ for } j=1,\ldots,m.
$$
The variational formulas allow to rewrite all the operators $\d_{\mathcal{B}_{\zeta,j}}$ in \eqref{operatorL}, for $j=1,\ldots,m$, in terms of derivatives with respect to the times $t_{\zeta,j}$ appearing as coefficients of the polar part of $\omega_{0,1}$. This is enough to rewrite the first line of \eqref{operatorL} as $L_{\infty}(x)$ in \eqref{inftyPart}.

For $L_{\Lambda}(x)$, we make use of \eqref{lambdaOp} to express the remaining operators $\d_{\mathcal{B}_{\zeta,m_l+1}}$ in terms of allowed times $t_{\zeta,j}$, $j=1,\ldots,m$, and derivatives with respect to the poles $\lambda_l$ of $R(x)$:
\beq\label{outerOp}
(m_l+1)t_{\zeta_l,m_l}\xi_l^{-1}(\d_{\mathcal{B}_{\zeta_l,m_l+1}}-\d_{\mathcal{B}_{\sigma(\zeta_l),m_l+1}}) = \xi_l^{-1} \left(\frac{\d}{\d \lambda_l}-\sum_{j=0}^{m_l-1}(j+1)t_{\zeta_l,j}\frac{\d}{\d t_{\zeta_l,j+1}}\right).
\eeq
The second type of terms of the RHS of \eqref{outerOp} cancels the terms with $k=0$ from second line of \eqref{operatorL}, yielding \eqref{finitePart}.

\end{proof}

We have just found a differential operator $L(x)$ in the times and the poles of $R(x)$, whose coefficients are rational functions of $x$, with poles at $x=\infty$ or $x=x(\zeta_i)$, i.e.~the same poles as $R(x)$, with at most the same degrees.

\begin{example} In the elliptic case of curves of the form $y^2=x^3+tx+V$ we have only one pole, at $\zeta_i=\infty$, of degree $m_i=5$, with $d_i=-2$. The only non-vanishing times are $t_{\infty,5}=-2$ and $t_{\infty,1}=-t$, and thus only the terms with $j=5$ and $k=0$ contribute:
\bea
L(x) 
&=&  \sum_{j=1,5} t_{\infty,j} \sum_{0\leq k \leq -2+(j-1)/2} x^k  (j/2-k-2) \frac{\partial}{\partial t_{\infty,j-2(k+2)}}   \nonumber \\
&=&   t_{\infty,5}  (5/2-0-2) \frac{\partial}{\partial t_{\infty,{5-2(0+2)}}}   \nonumber \\
&=&  -2 \   \Big(\frac52-2\Big) \frac{\partial}{\partial t_{\infty,{1}}}   \nonumber \\
&=&  - \frac{\partial}{\partial t_{\infty,{1}}}   \nonumber \\
&=& \frac{\partial}{\partial t}. \nonumber
\eea
\end{example}

\begin{example} In the Airy case, $y^2=x$, we have only one pole, at $\zeta_i=\infty$, of degree $m_i=3$, with $d_i=-2$. The sum is empty and
$$
L(x)=0.
$$
\end{example}

\begin{remark}
More generally, the admissible curves considered in \cite{BouchardEynard_QC} are those for which
$$
L(x)=0.
$$
\end{remark}

\section{PDEs quantizing any hyperelliptic curve}\label{section4}

For $r\geq 1$, let $D=\sum_{i=1}^r \alpha_i [p_i]$ be a divisor on $\Sigma$, with $p_i\in \Sigma$. We call $\sum_i \alpha_i$ the \emph{degree} of the divisor and denote ${\rm Div}_0(\Sigma)$ the set of divisors of degree $0$. For $D\in {\rm Div}_0(\Sigma)$, we define the integration of a $1$-form $\rho(z)$ on $\Sigma$ as
\beq
\int_D \rho(z) \coloneqq \sum_i\alpha_i\int_{o}^{p_i} \rho(z),
\eeq
where $o\in\Sigma$ is an arbitrary base point. This integral is well defined locally, meaning that it is independent of the base point $o$ because the degree of the divisor is zero, however it depends on a choice of homotopy class from $o$ to $p_i$.

For $(g,n)\neq(0,2)$ consider the functions of $D$, defined locally:
\begin{align*}
F_{g,0}(D)& \coloneqq  F_g = \omega_{g,0}, \\
F_{g,n}(D)& \coloneqq \overbrace{\int_D\cdots\int_D}^{n} \omega_{g,n}(z_1,\ldots,z_n), \\
F_{g,n}^{\prime}(z;D) & \coloneqq \frac{1}{\dd x(z)}\overbrace{\int_D\cdots\int_D}^{n-1} \omega_{g,n}(z,z_2\ldots,z_n), \\ 
F_{g,n}^{\prime\prime}(z,\tilde{z};D) & \coloneqq \frac{1}{\dd x(z)\dd x(\tilde{z})}\overbrace{\int_D\cdots\int_D}^{n-2} \omega_{g,n}(z,\tilde{z},z_3,\ldots,z_n).
\end{align*}
Recall that 
$$
B(z_1, z_2)\coloneqq \dd_1\dd_2\log \left(E(z_1,z_2)\sqrt{\dd x(z_1)\dd x(z_2)}\right),
$$
with $E(z_1,z_2)$ being the prime form, which is defined in \cite{EO07}, and satisfies that it vanishes only if $z_1=z_2$ with a simple zero and has no pole.

For $(g,n)=(0,2)$ define:
\begin{align}
F_{0,2}(D)& \coloneqq  2\sum_{i<j} \alpha_i\alpha_j \log \left( E(p_i,p_j) \sqrt{\dd x(p_i)\dd x(p_j)} \right) , \\
F_{0,2}^{\prime}(z;D) & \coloneqq \frac{1}{\dd x(z)} \dd_z \left( \sum_{i=1}^r \alpha_i  \log \left( E(z,p_i) \sqrt{\dd x(z)\dd x(p_i)} \right) \right) , \label{F02prime} \\ 
F_{0,2}^{\prime\prime}(z,\tilde{z};D) & \coloneqq \frac{B(z,\tilde z)}{\dd x(z)\dd x(\tilde{z})} .
\end{align}

Since the $\omega_{g,n}$ are symmetric, we have the following relations for $(g,n)\neq (0,2)$:
\begin{align}
\frac{\dd}{\dd x_i} F_{g,n}(D) & = n \alpha_i F_{g,n}^{\prime}(p_i;D), \label{1stsymm} \\
\left(\frac{\dd}{\dd x_i}\right)^2 F_{g,n}(D) & = n(n-1) \alpha_i^2 F_{g,n}^{\prime\prime}(p_i,p_i;D)+n\alpha_i \left(\frac{\dd}{\dd \tilde{x}} F_{g,n}^{\prime}(\tilde{p};D)\right)_{\tilde{p}=p_i}, \label{2ndsymm}
\end{align}
where $x_i \coloneqq x(p_i)$, $\tilde{x} \coloneqq x(\tilde{p})$, and $\frac{\dd}{\dd x}$ acts on meromorphic functions by taking exterior derivative and dividing by $\dd x$, which amounts to derivate an analytic expansion of the meromorphic function with respect to a local variable $x$. 

For $(g,n)=(0,2)$ we have
\begin{align}
\label{derivcyl} \frac{\dd}{\dd x_i} F_{0,2}(D) & = 2 \alpha_i \lim_{z\to p_i} \left( F_{0,2}^{\prime}(z;D) - \alpha_i \frac{\dd}{\dd x(z)} \log \left( E(z,p_i) \sqrt{\dd x(z)\dd x(p_i)} \right) \right)  \\
& = 2 \alpha_i \sum_{j\neq i} \alpha_j \frac{\dd}{\dd x(p_i)} \log \left( E(p_i,p_j) \sqrt{\dd x(p_i)\dd x(p_j)} \right)  . \nonumber
\end{align}

\begin{remark} Integrating the first part of Theorem~\ref{loop} over a divisor $D$ of degree $0$, we obtain:
\beq\label{cylprime}
F_{0,2}^{\prime}(z;D)+F_{0,2}^{\prime}(-z;D)=\sum_{i=1}^r\frac{\alpha_i}{x(z)-x(p_i)}.
\eeq
\end{remark}

\begin{lemma}\label{Lemmaintint02} Let $\zeta$ be a pole of $\omega_{0,1}$ and $\mathcal{B}_{\zeta,k}$ the $k$th 2nd type cycle around $\zeta$, with $k\geq 1$. Consider $t_{\zeta,k}$ the corresponding KP time. Then,
\beq\label{intint02}
\int_D\int_D\int_{\mathcal{B}_{\zeta,k}}\omega_{0,3}(z,z_1,z_2)=\int_D\int_D\frac{\partial}{\partial t_{\zeta,k}}\omega_{0,2}(z_1,z_2) =\frac{\partial}{\partial t_{\zeta,k}}F_{0,2}(D).
\eeq
\end{lemma}
\begin{proof}
\begin{multline}
\int_D\int_D\left(B(z_1,z_2)-\frac{\dd x(z_1)\dd x(z_2)}{(x(z_1)-x(z_2))^2}\right)+2\sum_{i<j}\alpha_i \alpha_j\log(x_i-x_j)= \\
2\sum_{i<j}\alpha_i\alpha_j\log \left( \frac{E(p_i,p_j)\sqrt{\dd x_i \dd x_j}}{x_i-x_j}\right) + 2\sum_{i<j} \alpha_i\alpha_j \log (x_i-x_j) + \sum_{i}\alpha_i^2 \log \frac{\dd x_i}{\dd x_i}=F_{0,2}(D).\label{02convention}
\end{multline}
Taking the derivative with respect to $t_{\zeta,k}$ of the first line gives the left hand side of \eqref{intint02} because we are taking this derivative at fixed $x$.
\end{proof}

Consider $\dd E(z,p_1)\coloneqq\dd_z\log \left(E(z,p_1)\sqrt{\dd x(z)\dd x(p_1)}\right)$ and observe that 
\beq\label{lim0}
\lim_{z\rightarrow p_1}\frac{\dd E(z,p_1)}{\dd x(z)}-\frac{1}{x(z)-x(p_1)} =0.
\eeq
With this notation one can rewrite \eqref{F02prime} as
$$
F_{0,2}^{\prime}(z;D)=\sum_{i=1}^r \alpha_i \frac{\dd E(z,p_i)}{\dd x(z)}
$$
and $\frac{\dd}{\dd x_i} F_{0,2}(D) = 2 \alpha_i \sum_{j\neq i} \alpha_j \frac{\dd E(p_i,p_j)}{\dd x(p_i)}$.

\vspace{0.2cm}
We define 
\beq
S_0(D)\coloneqq \int_{D} y \dd x = F_{0,1}(D), \;\; 
S_1(D)\coloneqq \log \prod_{i<j} \left(E(p_i,p_j)\sqrt{\dd x (p_i)\dd x(p_j)}\right)^{\alpha_i\alpha_j}=\frac{F_{0,2}(D)}{2}, \nonumber
\eeq
\beq
S_m(D)\coloneqq \sum_{\substack{2g-2+n =m-1\\ g\geq 0, n\geq 1}}\frac{F_{g,n}(D)}{n!}, \nonumber
\vspace{-0.3cm}
\eeq
and
\beq
\psi=\psi(D,\hbar)\coloneqq \exp (S(D,\hbar)), \text{ with } S(D,\hbar)\coloneqq \sum_{m=0}^{\infty}\hbar^{m-1}S_m(D).
\eeq

\begin{theorem} Let $F\coloneqq \sum_{g> 0}\hbar^{2g-2}F_g$. For every $k=1,\ldots,r$, assuming $\alpha_k^2=1$, we obtain
\beq\label{eqPDEgeneralL}
\hbar^2\Bigg(\frac{\dd^2}{\dd x_k^2} -\sum_{i\neq k} \frac{\frac{\dd}{\dd x_i} + \frac{\alpha_i}{\alpha_k}\frac{\dd}{\dd x_k}}{x_k-x_i}- L(x_k)-L(x_k).F+\sum_{\substack{i\neq j \\ i\neq k, j\neq k}} \frac{\alpha_i\alpha_j}{(x_k-x_i)(x_i-x_j)}\Bigg)\psi= R(x_k)\psi.
\eeq
\end{theorem}
\begin{proof}
We will give the proof of the claim for the case $k=1$, but it works exactly the same for every $k$.

Let us first consider the generic situation with $(g,n)\neq (0,0),(0,1),(0,2),(1,0)$. Using \eqref{loopeq} and \eqref{operatorLeq} for $n=0$, we can write $P_{g,0}(x)$, with $x=x(z)$, as 
\beq
F_{g-1,2}^{\prime\prime}(z,z;D)+\sum_{\substack{g_1+g_2=g}}F_{g_1,1}^{\prime}(z;D)F_{g_2,1}^{\prime}(z;D) = L(x).F_{g,0}(D) =L(x).F_{g}.
\eeq
Setting $z=p_1$, $x_1=x(p_1)$ and using \eqref{1stsymm} and \eqref{2ndsymm}, we obtain
\begin{align}\label{expression0}
& \frac{1}{2 \alpha_1^2} \left(\frac{\dd}{\dd x_1}\right)^2 F_{g-1,2}(D) 
- \frac{1}{\alpha_1} \left( \frac{\dd}{\dd  x(\tilde p)}  F_{g-1,2}^{\prime}(\tilde p;D)\right)_{\tilde p=p_1}  +
\sum_{\substack{g_1+g_2=g}} \frac{1}{\alpha_1^2} \frac{\dd F_{g_1,1}(D)}{\dd x_1} \frac{\dd F_{g_2,1}(D)}{\dd x_1} \\ & =L(x_1).F_{g}(D). \nonumber
\end{align}

For $n>0$, we integrate \eqref{loopeq} $n$ times over $D$ and use \eqref{operatorLeq} to get
\begin{multline}
F_{g-1,n+2}^{\prime\prime}(z,z;D)+\sum^{\text{no } (0,2)}_{\substack{g_1+g_2=g \\ n_1+n_2=n}}\binom{n}{n_1}F_{g_1,n_1+1}^{\prime}(z;D)F_{g_2,n_2+1}^{\prime}(z;D)+  \\
-2nF'_{0,2}(-z;D)F'_{g,n}(z;D)
-n\sum_{i=1}^r\alpha_i \frac{F_{g,n}^{\prime}(p_i;D)-F_{g,n}^{\prime}(z;D)}{x(z)-x(p_i)} = L(x).F_{g,n}(D), \nonumber
\end{multline}
where the sum of the second term is taken over $g_i, n_i\geq 0$ and ``no $(0,2)$'' means we exclude the cases with $(g_i,n_i)=(0,1)$ for $i=1,2$, for which we have used that
$$
(F^{\prime}_{0,2}(z;D)-F_{0,2}^{\prime}(-z;D))F_{g,n}^{\prime}(z;D)=-2 F'_{0,2}(-z;D)F'_{g,n}(z;D)+\sum_{i=1}^r\alpha_i\frac{F_{g,n}^{\prime}(z;D)}{x(z)-x(p_i)},
$$
which follows from~\eqref{cylprime}.

Letting $z=p_1$, $x_i=x(p_i)$, dividing by $n!$ and using \eqref{1stsymm} and \eqref{2ndsymm}, we obtain
\begin{multline}
\frac{1}{(n+2)! \alpha_1^2} \left(\frac{\dd}{\dd x_1}\right)^2 F_{g-1,n+2}(D)
- \frac{1}{(n+1)! \alpha_1} \left( \frac{\dd}{\dd  x(\tilde p)}  F_{g-1,n+2}^{\prime}(\tilde p;D)\right)_{\tilde p=p_1}  + \\
\sum^{\text{no } (0,2)}_{\substack{g_1+g_2=g \\ n_1+n_2=n}} \frac{1}{(n_1+1)! (n_2+1)! \alpha_1^2} \frac{\dd F_{g_1,n_1+1}(D)}{\dd x_1} \frac{\dd F_{g_2,n_2+1}(D)}{\dd x_1}   \\
- \frac{1}{n!} \sum_{i=2}^r \frac{1}{x(p_1)-x(p_i)}  \left( \frac{\dd F_{g,n}(D)}{\dd x_i} - \frac{\alpha_i}{\alpha_1}\frac{\dd F_{g,n}(D)}{\dd x_1}\right)+\frac{\alpha_1}{(n-1)!} \left( \frac{\dd F'_{g,n}(\tilde p;D)}{\dd x(\tilde p)}\right)_{\tilde p=p_1}   \\
- \frac{2}{n!\alpha_1} \frac{\dd F_{g,n}(D)}{\dd x_1} \ F'_{0,2}(-p_1;D)
= L(x_1).\frac{F_{g,n}(D)}{n!}. \nonumber
\end{multline}
Using \eqref{cylprime} again, we obtain the following expression for the left hand side:
\begin{multline}
\frac{1}{(n+2)! \alpha_1^2} \left(\frac{\dd}{\dd x_1}\right)^2 F_{g-1,n+2}(D)
- \frac{1}{(n+1)! \alpha_1} \left( \frac{\dd}{\dd  x(\tilde p)}  F_{g-1,n+2}^{\prime}(\tilde p;D)\right)_{\tilde p=p_1}  + \\
\sum^{\text{no } (0,2)}_{\substack{g_1+g_2=g \\ n_1+n_2=n}} \frac{1}{(n_1+1)! (n_2+1)! \alpha_1^2} \frac{\dd F_{g_1,n_1+1}(D)}{\dd x_1} \frac{\dd F_{g_2,n_2+1}(D)}{\dd x_1}   \\
+\frac{\alpha_1}{(n-1)!} \left( \frac{\dd F'_{g,n}(\tilde p;D)}{\dd x(\tilde p)}\right)_{\tilde p=p_1}
- \frac{1}{n!} \sum_{i=2}^r \frac{1}{x(p_1)-x(p_i)}  \left( \frac{\dd F_{g,n}(D)}{\dd x_i} + \frac{\alpha_i}{\alpha_1}\frac{\dd F_{g,n}(D)}{\dd x_1}\right)   \\
- \frac{2}{n!\alpha_1} \frac{\dd F_{g,n}(D)}{\dd x_1} \left(  F'_{0,2}(\tilde p;D)- \frac{\alpha_1}{x(\tilde p)-x(p_1)}  \right)_{\tilde p=p_1}\\
=\frac{1}{(n+2)! \alpha_1^2} \left(\frac{\dd}{\dd x_1}\right)^2 F_{g-1,n+2}(D)
- \frac{1}{(n+1)! \alpha_1} \left( \frac{\dd}{\dd  x(\tilde p)}  F_{g-1,n+2}^{\prime}(\tilde p;D)\right)_{\tilde p=p_1}   \\
+\frac{\alpha_1}{(n-1)!} \left( \frac{\dd F'_{g,n}(\tilde p;D)}{\dd x(\tilde p)}\right)_{\tilde p=p_1} +\sum_{\substack{g_1+g_2=g \\ n_1+n_2=n}} \frac{1}{(n_1+1)! (n_2+1)! \alpha_1^2} \frac{\dd F_{g_1,n_1+1}(D)}{\dd x_1} \frac{\dd F_{g_2,n_2+1}(D)}{\dd x_1}   \\
- \frac{1}{n!} \sum_{i=2}^r \frac{1}{x(p_1)-x(p_i)}  \left( \frac{\dd F_{g,n}(D)}{\dd x_i} + \frac{\alpha_i}{\alpha_1}\frac{\dd F_{g,n}(D)}{\dd x_1}\right),\\
\end{multline}
where for the last equality we have used \eqref{derivcyl} and \eqref{lim0}.

For all $\ell \geq 3$, we sum this expression for all $g\geq 0, n \geq 1$ such that $2g-2+n=\ell-2$ and the corresponding $P_{g,0}(x(z_1))$ for $n=0$ from \eqref{expression0} for all $g\geq 0$ such that $2g-2=\ell-2$ to obtain:
\begin{multline}
\sum_{\substack{2g+n =\ell\\ g\geq 0, n\geq 0}} \Bigg(\frac{1}{(n+2)! \alpha_1^2} \left(\frac{\dd}{\dd x_1}\right)^2 F_{g-1,n+2}(D) \nonumber\\
- \frac{1}{(n+1)! \alpha_1} \left( \frac{\dd}{\dd  x(\tilde p)}  F_{g-1,n+2}^{\prime}(\tilde p;D)\right)_{\tilde p=p_1}+\frac{\alpha_1}{(n-1)!} \left( \frac{\dd F'_{g,n}(\tilde p;D)}{\dd x(\tilde p)}\right)_{\tilde p=p_1}
\nonumber\\ - \frac{1}{n!} \sum_{i=2}^r \frac{1}{x(p_1)-x(p_i)}  \left( \frac{\dd F_{g,n}(D)}{\dd x_i} + \frac{\alpha_i}{\alpha_1}\frac{\dd F_{g,n}(D)}{\dd x_1}\right)\Bigg)  \nonumber \\
+ \frac{1}{\alpha_1^2} \sum_{\ell_1+\ell_2=\ell} 
\Bigg(\sum_{\substack{2g_1-2+n_1=\ell_1-1 \\ g_1\geq 0,n_1\geq 1}} \frac{1}{n_1!} \frac{\dd F_{g_1,n_1}(D)}{\dd x_1} 
\sum_{\substack{2g_2-2+n_2=\ell_2-1 \\ g_1\geq 0,n_2\geq 1}}\frac{1}{n_2!}\frac{\dd F_{g_2,n_2}(D)}{\dd x_1}\Bigg) = \sum_{\substack{2g+n =\ell\\ g\geq 0, n\geq 0}} L(x_1). \frac{F_{g,n}}{n!}. \nonumber \\
\end{multline}

In this sum over topologies, assuming that $\frac{1}{\alpha_1}=\alpha_1$, all the first terms of the second line cancel against all the second terms except the ones with $n=1$, which complete the terms of the first line, using \eqref{2ndsymm} for the special case $n=1$ in which only the second term of the right hand side survives.

Therefore, for $\ell \geq 3$, we have proved
$$
\left(\frac{\dd}{\dd x_1}\right)^2 S_{\ell-1} +\frac{1}{\alpha_1^2}\sum_{\ell_1+\ell_2=\ell}\frac{\dd}{\dd x_1}S_{\ell_1}\frac{\dd}{\dd x_1}S_{\ell_2} -\sum_{i=2}^r  \frac{\frac{\dd S_{\ell-1}}{\dd x_i} +\frac{\alpha_i}{\alpha_1} \frac{\dd S_{\ell-1}}{\dd x_1}}{x(p_1)-x(p_i)}
$$
$$
= L(x_1).S_{\ell-1}+\begin{cases}L(x_1).F_{\ell/2}, & \ell \text{ even,}\\ 0, & \ell \text{ odd,}\end{cases}
$$
that is
\begin{multline}
[\hbar^{\ell}]\left[\hbar^2\left(\left(\frac{\dd}{\dd x_1}\right)^2 S +  \frac{1}{\alpha_1^2}\frac{\dd}{\dd x_1}S\frac{\dd}{\dd x_1}S -\sum_{i=2}^r  \frac{\frac{\dd S}{\dd x_i} +\frac{\alpha_i}{\alpha_1} \frac{\dd S}{\dd x_1}}{x(p_1)-x(p_i)}\right)\right] = \\
[\hbar^{\ell}]\left[ \hbar^2 (L(x_1).S+L(x_1).F)+R(x_1)\right].
\end{multline}

Let us finally consider the special cases before the assumption $\alpha_1^2=1$:

$\bullet$ For $(g,n)=(0,0)$ we get
\beq
\frac{1}{\alpha_1^2}  \left(\frac{\dd F_{0,1}(D)}{\dd x_1}\right)^2 = R(x_1), \quad \text{i.e. } \quad \frac{1}{\alpha_1^2}  \left(\frac{\dd}{\dd x_1}S_0\right)^2 = R(x_1), \quad \text{ or }
\eeq
$$
[\hbar^0]\bigg(\hbar^2\bigg(\bigg(\frac{\dd}{\dd x_1}\bigg)^2 S + \frac{1}{\alpha_1^2}\frac{\dd}{\dd x_1}S\frac{\dd}{\dd x_1}S -\sum_{i=2}^r  \frac{\frac{\dd S}{\dd x_i} +\frac{\alpha_i}{\alpha_1} \frac{\dd S}{\dd x_1}}{x(p_1)-x(p_i)}
-  L(x_1).S- L(x_1).F\bigg)-R(x_1)\bigg)=0.
$$

$\bullet$ For $(g,n)=(0,1)$ we get
\beq
2F'_{0,1}(z;D)F'_{0,2}(z;D) -\sum_{i=1}^r\alpha_i \frac{F_{0,1}^{\prime}(p_i;D)+F_{0,1}^{\prime}(z;D)}{x(z)-x(p_i)} \ 
= L(x).F_{0,1}(D). \nonumber
\eeq
Thus
\begin{multline}
2F'_{0,1}(z;D)F'_{0,2}(z;D)-2\alpha_1 \frac{F_{0,1}^{\prime}(z;D)}{x(z)-x(p_1)} 
-\alpha_1 \frac{F_{0,1}^{\prime}(p_1,D)-F_{0,1}^{\prime}(z;D)}{x(z)-x(p_1)} 
\\ 
-\sum_{i=2}^r \alpha_i \frac{F_{0,1}^{\prime}(p_i;D)+F_{0,1}^{\prime}(z;D)}{x(z)-x(p_i)} \ 
= L(x).F_{0,1}(D). \nonumber
\end{multline}
At $z=p_1$ this gives
\begin{multline}
2F'_{0,1}(p_1;D) \left( F'_{0,2}(z;D)-\alpha_1 \frac{1}{x(z)-x(p_1)} \right)_{z=p_1}
+\alpha_1  \left( \frac{\dd F_{0,1}^{\prime}(z;D)}{\dd x}\right)_{z=p_1}
\\ 
-\sum_{i=2}^r  \frac{1}{x(p_1)-x(p_i)} \left( \frac{\dd F_{0,1}(D)}{\dd x_i} +\frac{\alpha_i}{\alpha_1} \frac{\dd F_{0,1}(D)}{\dd x_1} \right)\ 
= L(x_1).F_{0,1}(D), \nonumber
\end{multline}
and equivalently:
$$
\bigg(\frac{\dd}{\dd x_1}\bigg)^2 F_{0,1}(D) + \frac{1}{\alpha_1^2}\frac{\dd}{\dd x_1}F_{0,1}(D)\frac{\dd}{\dd x_1}F_{0,2}(D) -\sum_{i=2}^r  \frac{\frac{\dd F_{0,1}(D)}{\dd x_i} +\frac{\alpha_i}{\alpha_1} \frac{\dd F_{0,1}(D)}{\dd x_1}}{x(p_1)-x(p_i)}
= L(x_1).F_{0,1}(D),
$$

$$
\bigg(\frac{\dd}{\dd x_1}\bigg)^2 S_0 + \frac{2}{\alpha_1^2}\frac{\dd}{\dd x_1}S_0\frac{\dd}{\dd x_1}S_1 -\sum_{i=2}^r  \frac{\frac{\dd S_0}{\dd x_i} +\frac{\alpha_i}{\alpha_1} \frac{\dd S_0}{\dd x_1}}{x(p_1)-x(p_i)}
= L(x_1).S_0,
$$

$$
[\hbar]\bigg(\hbar^2\bigg(\bigg(\frac{\dd}{\dd x_1}\bigg)^2 S + \frac{1}{\alpha_1^2}\frac{\dd}{\dd x_1}S\frac{\dd}{\dd x_1}S -\sum_{i=2}^r  \frac{\frac{\dd S}{\dd x_i} +\frac{\alpha_i}{\alpha_1} \frac{\dd S}{\dd x_1}}{x(p_1)-x(p_i)}
- L(x_1).S-L(x_1).F\bigg)-R(x_1)\bigg)=0.
$$

$\bullet$ For $(g,n)=(0,2)$ we first rewrite \eqref{loopeq}:
\begin{multline}
P_{0,2}(x(z),z_1,z_2)-2y(z)\frac{\omega_{0,3}(z,z_1,z_2)}{\dd x}=\\
-\frac{B(z,z_1)B(-z,z_2)}{(\dd x)^2}-\frac{B(-z,z_1)B(z,z_2)}{(\dd x)^2}+\dd_1\frac{B(z_2,-z_1)}{(x-x_1)\dd x_1}+\dd_2\frac{B(z_1,-z_2)}{(x-x_2)\dd x_2}= \\
 2\frac{B(z,z_1)B(z,z_2)}{(\dd x)^2} + \dd_1 \frac{1}{x-x_1}\bigg(\frac{B(z_2,-z_1)}{\dd x_1}-\frac{B(z,z_2)}{\dd x}\bigg)+ \dd_2 \frac{1}{x-x_2}\bigg(\frac{B(z_1,-z_2)}{\dd x_2}-\frac{B(z,z_1)}{\dd x}\bigg)= \\
 2\frac{B(z,z_1)B(z,z_2)}{(\dd x)^2} - \dd_1 \frac{1}{x-x_1}\bigg(\frac{B(z_2,z_1)}{\dd x_1}+\frac{B(z,z_2)}{\dd x}\bigg)- \dd_2 \frac{1}{x-x_2}\bigg(\frac{B(z_1,z_2)}{\dd x_2}+\frac{B(z,z_1)}{\dd x}\bigg) \\
+ \dd_1\dd_2 \frac{1}{x-x_1}\frac{1}{x-x_2}.
\end{multline}
Now we integrate twice over $D$:
\begin{multline}\label{intintP}
\int_D\int_D P_{0,2}(x(z),z_1,z_2)-2F_{0,1}'(z;D)F_{0,3}'(z;D)=\\
2 (F'_{0,2}(z;D))^2- 2\sum_{i=1}^r \alpha_i \frac{1}{x-x_i}F_{0,2}'(z;D)-\sum_{i=1}^r\frac{2\alpha_i}{x-x_i}\sum_{j\neq i} \alpha_j\bigg(\frac{\dd E(p_i,p_j)}{\dd x_i}-\frac{1}{x_i-x_j}\bigg).
\end{multline}
We introduce $\hat{F}_{0,2}'(z;D)=F_{0,2}'(z;D)-\alpha_1\frac{\dd E(z,p_1)}{\dd x}$ and obtain:
\begin{multline}
2 (\hat{F}'_{0,2}(z;D))^2+ 4\alpha_1\hat{F}_{0,2}'(z;D)\frac{\dd E(z,p_1)}{\dd x} + 2\alpha_1^2\frac{(\dd E(z,p_1))^2}{(\dd x)^2}- \frac{2\alpha_1}{x-x_1}\hat{F}_{0,2}'(z;D) - \\
  \frac{2\alpha_1^2}{x-x_1}\frac{\dd E(z,p_1)}{\dd x} 
- 2\sum_{i\neq 1}\frac{\alpha_i}{x-x_i}\hat{F}_{0,2}'(z;D) - 
2\sum_{i\neq 1} \frac{\alpha_1\alpha_i}{x-x_i}\frac{\dd E(z,p_1)}{\dd x}-\\
 \frac{1}{x-x_1}\frac{\dd}{\dd x_1} F_{0,2}(D) -\sum_{i\neq 1} \frac{1}{x-x_i}\frac{\dd}{\dd x_i} F_{0,2}(D)+ 2\sum_{\substack{i\neq j \\ i\neq 1}}\frac{\alpha_i\alpha_j}{(x-x_i)(x_i-x_j)}+ 2\sum_{j\neq 1}\frac{\alpha_1\alpha_j}{(x-x_1)(x_1-x_j)}.
\end{multline}

Observe that 
\begin{equation}\label{limit}
\lim_{z\rightarrow p_1}\bigg(-\sum_{i\neq 1} \frac{\alpha_1\alpha_i}{x-x_i}\frac{\dd E(z,p_1)}{\dd x}+\sum_{j\neq 1}\frac{\alpha_1\alpha_j}{(x-x_1)(x_1-x_j)}\bigg)=\sum_{i\neq 1}\frac{\alpha_1\alpha_i}{(x_1-x_i)^2}.
\end{equation}

Using this, at $z=p_1$ we obtain
\begin{multline}
\frac{1}{2\alpha_1^2} \bigg(\frac{\dd F_{0,2}(D)}{\dd x_1}\bigg)^2+ \bigg(\frac{\dd}{\dd x_1}\bigg)^2 F_{0,2}(D)- \sum_{i\neq 1}\frac{1}{x_1-x_i}\bigg(\frac{\dd F_{0,2}(D)}{\dd x_i}+\frac{\alpha_i}{\alpha_1}\frac{\dd F_{0,2}(D)}{\dd x_1}\bigg) + \\
2\alpha_1^2\mathcal{S}(p_1) + 2\sum_{\substack{i\neq j \\ i\neq 1, j\neq 1}} \frac{\alpha_i\alpha_j}{(x_1-x_i)(x_i-x_j)},
\end{multline}
where we have called $\mathcal{S}(p_1)$ the limit
$$
\lim_{z\rightarrow p_1}\frac{\dd E(z,p_1)}{\dd x(z)}\bigg(\frac{\dd E(z,p_1)}{\dd x(z)}-\frac{1}{x(z)-x(p_1)}\bigg).
$$
Using \eqref{operatorLeq} in this special case, together with Lemma~\ref{Lemmaintint02}, we obtain
\beq
\int_D\int_D P_{0,2}(x;z_1,z_2)= \int_D\int_D L(x).\omega_{0,2}(z_1,z_2) = L(x).F_{0,2}(D).
\eeq
For the first term of \eqref{intintP}, we thus obtain:
\begin{multline}
L(x_1).\frac{F_{0,2}(D)}{2} = \frac{1}{3\alpha_1^2}\frac{\dd}{\dd x_1}F_{0,1}(D)\frac{\dd}{\dd x_1}F_{0,3}(D)+\frac{1}{4\alpha_1^2} \left(\frac{\dd F_{0,2}(D)}{\dd x_1}\right)^2+ \frac{1}{2}\left(\frac{\dd}{\dd x_1}\right)^2 F_{0,2}(D)+\\
\alpha_1^2\mathcal{S}(p_1)-
 \frac{1}{2}\sum_{i\neq 1}\frac{1}{x_1-x_i}\left(\frac{\dd F_{0,2}(D)}{\dd x_i}+\frac{\alpha_i}{\alpha_1}\frac{\dd F_{0,2}(D)}{\dd x_1}\right) + \sum_{\substack{i\neq j \\ i\neq 1, j\neq 1}} \frac{\alpha_i\alpha_j}{(x_1-x_i)(x_i-x_j)}.
\end{multline}

$\bullet$ For $(g,n)=(1,0)$ we get
$$
\frac{-B(z,-z)}{\dd x(z)^2} + 
2F'_{0,1}(z;D)F'_{1,1}(z;D) \ 
= L(x_1).F_{1,0}(D). \nonumber
$$
At $z=p_1$ this gives
$$
\frac{-B(p_1,-p_1)}{\dd x_1^2}+
2 \frac{1}{\alpha_1^2} \frac{\dd F_{0,1}(D)}{\dd x_1} \frac{\dd F_{1,1}(D)}{\dd x_1} 
= L(x_1).F_{1,0}(D).\nonumber
$$
Using that $\frac{B(p_1,-p_1)}{\dd x_1^2}= \mathcal{S}(p_1)$ and summing the expressions for $(0,2)$ and $(1,0)$, we obtain:
\begin{multline}
\bigg(\frac{\dd}{\dd x_1}\bigg)^2 S_{1} +\frac{1}{\alpha_1^2}\bigg(\frac{\dd}{\dd x_1}S_{1}\bigg)^2 +\frac{2}{\alpha_1^2}\frac{\dd}{\dd x_1}S_{0}\frac{\dd}{\dd x_1}S_{2} -\sum_{i=2}^r  \frac{\frac{\dd S_{1}}{\dd x_i} +\frac{\alpha_i}{\alpha_1} \frac{\dd S_{1}}{\dd x_1}}{x(p_1)-x(p_i)} \\
+(\alpha_1^2-1)\mathcal{S}(p_1) + \sum_{\substack{i< j \\ i\neq 1}} \frac{\alpha_i\alpha_j}{(x_1-x_i)(x_1-x_j)} = L(x_1).S_{1}+L(x_1).F_{1},
\end{multline}
that is
\begin{multline}
[\hbar^{2}]\bigg[\hbar^2\bigg(\bigg(\frac{\dd}{\dd x_1}\bigg)^2 S +  \frac{1}{\alpha_1^2}\frac{\dd}{\dd x_1}S\frac{\dd}{\dd x_1}S -\sum_{i=2}^r  \frac{\frac{\dd S}{\dd x_i} +\frac{\alpha_i}{\alpha_1} \frac{\dd S}{\dd x_1}}{x(p_1)-x(p_i)} +(\alpha_1^2-1)\mathcal{S}(p_1) +(\star)\bigg)\bigg] = \\
[\hbar^{2}]\bigg[ \hbar^2 (L(x_1).S+ L(x_1).F)+R(x_1) \bigg],
\end{multline}
with 
\beq\label{star}
(\star)=\sum_{\substack{i\neq j \\ i\neq 1, j\neq 1}} \frac{\alpha_i\alpha_j}{(x_1-x_i)(x_i-x_j)}.
\eeq
From the assumption $\alpha_1 =\frac{1}{\alpha_1}$, and summing over all topologies, we get the claim.
\end{proof}

\begin{remark}
Very often in the literature a different convention is used to regularize the $(0,2)$ term of the wave function:
$$
\tilde{\psi}(D,t,\hbar)\coloneqq \exp{\bigg(\tilde{S}_1(D,t)+\sum_{m\geq 0,m\neq 1} \hbar^{m-1}S_m(D,t)\bigg)},
$$
where $\tilde{S}_1(D,t)\coloneq \frac{1}{2}\int_D\int_D\big(B(z_1,z_2)-\frac{\dd x(z_1)\dd x(z_2)}{(x(z_1)-x(z_2))^2}\big)$. Using \eqref{02convention}, we obtain that the relation to our wave function is the following
$$
\psi(D,t,\hbar)=\tilde{\psi}(D,t,\hbar) \cdot \prod_{i<j}(x_i-x_j)^{\alpha_i\alpha_j}.
$$
\end{remark}

\subsection{PDE for Airy curve}

In this particular case, we had $P_{g,n}=0$.

Therefore, in this case we obtain the following system of PDEs:
\beq
\hbar^2\Bigg(\frac{\dd^2}{\dd x_k^2} -\sum_{i\neq k} \frac{\frac{\dd}{\dd x_i} + \frac{\alpha_i}{\alpha_k}\frac{\dd}{\dd x_k}}{x_k-x_i}+\sum_{\substack{i\neq j \\ i\neq k, j\neq k}} \frac{\alpha_i\alpha_j}{(x_k-x_i)(x_i-x_j)}\Bigg)\psi= x_k\psi,
\eeq
for every $k=1,\ldots,r$.

\begin{example}
Considering the divisor $D=[z_1]-[z_2]$, sending $z_2\rightarrow \infty$ and regularizing the $(0,1)$ factor of the wave function, we recover the Airy quantum curve from our PDE for $k=1$, with $x=x_1$:
$$
\bigg(\hbar^2\frac{\dd^2}{\dd x^2}-x\bigg)\psi=0.
$$
\end{example}

\subsection{PDE for Painlev\'e case} In this case we have $P_{g,n}=\frac{\partial}{\partial t}\omega_{g,n}(z_1,\ldots,z_n)$. 

Therefore, we obtain the following PDE:
$$
\hbar^2\bigg(\frac{\dd^2}{\dd x_k^2} -\sum_{i\neq k} \frac{\frac{\dd}{\dd x_i} + \frac{\alpha_i}{\alpha_k}\frac{\dd}{\dd x_k}}{x_k-x_i}- \frac{\partial}{\partial t}- \frac{\partial}{\partial t} F+(\star)\bigg)\psi(D)= (x_k^3+tx_k+V)\psi=\big(x_k^3+tx_k+\frac{\partial}{\partial t}\omega_{0,0}\big)\psi, 
$$
for every $k=1,\ldots,r$, where $(\star)$ is given by $\eqref{star}$.

\subsection{Reduced equation}

\medskip

Consider a divisor $D=[z]-[z']$ with 2 points, and call $x=x(z), x'=x(z')$.
The equation we have obtained is a PDE: it involves both $\dd/\dd x$ and $\dd/\dd x'$, as well as partial derivatives with respect to times when $L(x)\neq 0$. 
Let us show here that it is possible to eliminate $\dd/\dd x'$ and arrive to an equation involving only $\dd/\dd x$, as well as possibly times derivatives.

Define
\beq
\tilde\psi(z,z') \coloneqq (x-x')\psi([z]-[z'],t,\hbar) e^{F}.
\eeq
Define the differential operators
\bea
\mathcal D & \coloneqq & \hbar^2\frac{\dd^2}{\dd x^2}- \hbar^2L(x)-R(x), \label{operatorD} \\
\mathcal D' & \coloneqq & \hbar^2\frac{\dd^2}{\dd x'^2}- \hbar^2L(x')-R(x').
\eea
Equation \eqref{eqPDEgeneralL}  is equivalent to
\beq\label{eqKD}
\mathcal D \tilde\psi = \frac{\hbar^2}{x-x'} \left(\frac{\dd}{\dd x}+\frac{\dd}{\dd x'}\right)\tilde\psi = -\mathcal D' \tilde\psi.
\eeq
In particular this implies
\beq
\hbar^2 \frac{\dd}{\dd x'}\tilde\psi = -\hbar^2\frac{\dd}{\dd x}\tilde\psi+(x-x')\mathcal D \tilde\psi,
\eeq
and applying $\dd/\dd x'$ again we find
\bea
\hbar^2 \frac{\dd^2}{\dd x'^2}\tilde\psi 
&=& -\mathcal D \tilde\psi -\hbar^2 \frac{\dd}{\dd x}\frac{\dd}{\dd x'}\tilde\psi+(x-x')\mathcal D \frac{\dd}{\dd x'} \tilde\psi \nonumber \\
&=&\hbar^2 \frac{\dd^2}{\dd x^2}\tilde\psi -2\mathcal D \tilde\psi -(x-x') \Big( \frac{\dd}{\dd x}\mathcal D+\mathcal D\frac{\dd}{\dd x}-\hbar^{-2}\mathcal D(x-x')\mathcal D  \Big) \tilde\psi . \nonumber
\eea
Therefore 
\bea
0 &=& (\mathcal D+\mathcal D')\tilde\psi  \nonumber \\
&=& \mathcal D \tilde\psi +\Big(\hbar^2 \frac{\dd^2}{\dd x'^2} - R(x')- \hbar^2L(x')\Big)\tilde\psi \nonumber \\
&=& \bigg(\hbar^2 \frac{\dd^2}{\dd x^2}-\mathcal D - R(x')- \hbar^2L(x') - (x-x') \Big( \frac{\dd}{\dd x}\mathcal D+\mathcal D\frac{\dd}{\dd x}-\hbar^{-2} \mathcal D(x-x')\mathcal D  \Big)\bigg)\tilde\psi \nonumber \\
&=& \bigg(R(x)-R(x')+ \hbar^2(L(x)-L(x'))   -(x-x') \Big( \frac{\dd}{\dd x}\mathcal D+\mathcal D\frac{\dd}{\dd x}-\hbar^{-2} \mathcal D(x-x')\mathcal D  \Big)\bigg) \tilde\psi \nonumber \\
&=& \bigg(R(x)-R(x')+ \hbar^2(L(x)-L(x')) -(x-x') \Big( -\frac{\dd}{\dd x}\mathcal D+\mathcal D\frac{\dd}{\dd x}-\hbar^{-2} (x-x')\mathcal D^2  \Big)\bigg) \tilde\psi \nonumber \\
&=& \bigg(R(x)-R(x')+ \hbar^2(L(x)-L(x')) -(x-x') \Big( \frac{\dd R(x)}{\dd x}+\hbar^2 \frac{\dd L(x)}{\dd x}-\hbar^{-2} (x-x')\mathcal D^2  \Big)\bigg) \tilde\psi, \nonumber
\eea
and thus
\beq
\frac{R(x)-R(x')}{x-x'}\tilde\psi+ \hbar^2\frac{L(x)-L(x')}{x-x'}\tilde\psi =\bigg( \frac{\dd R(x)}{\dd x}+ \hbar^2\frac{\dd L(x)}{\dd x}-\hbar^{-2}(x-x')\mathcal D^2  \bigg) \tilde\psi.
\eeq
Finally,
\beq
\mathcal D^2 \tilde\psi = \frac{\hbar^2}{x-x'}\bigg(\frac{R(x)-R(x')}{x-x'}+ \hbar^2\frac{L(x)-L(x')}{x-x'} - \frac{\dd R(x)}{\dd x}- \hbar^2\frac{\dd L(x)}{\dd x}\bigg) \tilde\psi.
\eeq
This equation is a PDE, with rational coefficients $\in \mathbb C(x)$, involving $\dd/\dd x$ and $\partial/\partial t_k$s but no $\dd/\dd x'$ anymore.

Notice that the right hand side is of order $O(\hbar^2)$ in the limit $\hbar\to 0$, and $\mathcal D\to \hat y^2-R(x)$, where $\hat y=\hbar \dd/\dd x$.

\section{Quantum curves}\label{section5}

The goal now is to prove that $\psi(D)$ obeys an isomonodromic system of differential equations, and in particular this implies the existence of a quantum curve $\hat P(\hat{x},\hat y,\hbar)$ that annihilates $\psi(D)$.
To this purpose, we first prove that $\psi([z]-[z'])$ coincides with the integrable kernel of an isomonodromic system.
The way to prove it generalizes the method of \cite{bergre2009determinantal}, i.e.~first proving that the ratio of $\psi$ and the integrable kernel has to be a formal series of the form $1+O(z'^{-1})$ and then showing that the only solution of equations \eqref{eqPDEgeneralL} which has that behavior implies that the ratio must be $1$.

\subsection{Painlev\'e I (genus 0 case)}

We shall prove that $\psi([z]-[z'])$ coincides with the integrable kernel associated to the Painlev\'e I kernel.

Consider a solution of the Painlev\'e system \eqref{eqPainleveLax}, written as
$$
\Psi(x) = \begin{pmatrix}
A(x) & B(x) \cr
\tilde A(x) & \tilde B(x)
\end{pmatrix},
\quad
\det \Psi(x)=1,
$$
i.e.~$\Psi(x)$ satisfying \eqref{eqPainleveLax}:
$$
\bigg(\hbar \frac{\partial}{\partial x}- \mathcal L(x,t;\hbar)\bigg)\Psi(x) = 0 
\; , \quad
\bigg(\hbar \frac{\partial}{\partial t}- \mathcal R(x,t;\hbar)\bigg)\Psi(x) = 0.
$$
Define $A(x), \tilde A(x),B(x),\tilde B(x) $ as WKB $\hbar$-formal series solutions, with leading orders
\bea
A(x) &\sim & \frac{i}{\sqrt{2z}}\, e^{\hbar^{-1}\int_0^{x} y\dd x} (1 + O(\hbar)), \nonumber \\
B(x) &\sim & \frac{i}{\sqrt{2z}}\, e^{-\hbar^{-1}\int_0^{x} y\dd x} (1 + O(\hbar)), \nonumber \\
\tilde A(x) &\sim & i\sqrt{\frac{z}{2}}\, e^{\hbar^{-1}\int_0^{x} y\dd x} (1 + O(\hbar) ),\nonumber \\
\tilde B(x) &\sim & -i\sqrt{\frac{z}{2}}\, e^{-\hbar^{-1}\int_0^{x} y\dd x}  (1+ O(\hbar)), \nonumber
\eea
and with each coefficient of higher powers of $\hbar$ in $(1+O(\hbar))$ being a polynomial of $1/z$ that tends to 0 as $z\to\infty$. The choice of normalisation constants for the WKB solutions $C_A=C_B=\frac{i}{\sqrt{2}}$ is made such that $\det \Psi =-2C_A C_B=1$ and $C_A= C_B$.
The integrable kernel is defined as (a WKB formal series of $\hbar$):
\beq
K(x,x')\coloneqq\frac{A(x)\tilde B(x')-\tilde A(x)B(x')}{x-x'}.
\eeq
From the isomonodromic system \eqref{eqPainleveLax}, one can verify that this kernel obeys the same equation \eqref{eqPDEgeneralL} as $\psi([z]-[z'])$.
In fact, they are equal:

\begin{theorem}

We have, as formal WKB power series of $\hbar$
\beq
\psi([z]-[z'],t,\hbar) = \frac{A(x)\tilde B(x')-\tilde A(x)B(x')}{x-x'},
\eeq
where $x=x(z)$ and $x'=x(z')$.

\end{theorem}

\begin{proof}

Define the ratio
\beq
H(z,z') \coloneq \frac{(x-x')\psi([z]-[z'],t,\hbar)}{A(x)\tilde B(x')-\tilde A(x)B(x')}.
\eeq
It is a formal series of $\hbar$ whose coefficients are rational functions  of $z$ and $z'$. The leading orders show that
$$
H(z,z') = 1+O(\hbar).
$$
More explicitly, we have
\bea
H(z,z') & = & \frac{(z^2-z'^2)\frac{1}{(z-z')\sqrt{2z2z'}}\sum\hbar^kc_k\big(\frac{1}{z},\frac{1}{z'}\big)}{\frac12\sqrt{\frac{z'}{z}}\Big(\sum\hbar^k\alpha_k\big(\frac{1}{z}\big)\Big)\Big(\sum\hbar^{\ell}\tilde{\beta}_{\ell}\big(\frac{1}{z'}\big)\Big)+\frac12\sqrt{\frac{z}{z'}}\Big(\sum\hbar^k\tilde{\alpha}_k\big(\frac{1}{z}\big)\Big)\Big(\sum\hbar^{\ell}\beta_{\ell}\big(\frac{1}{z'}\big)\Big)} \nonumber \\
& = &\frac{\sum\hbar^kc_k\big(\frac{1}{z},\frac{1}{z'}\big)}{\bigg(\frac{1}{z}\Big(\sum\hbar^k\alpha_k\big(\frac{1}{z}\big)\Big)\Big(\sum\hbar^{\ell}\tilde{\beta}_{\ell}\big(\frac{1}{z'}\big)\Big)+\frac{1}{z'}\Big(\sum\hbar^k\tilde{\alpha}_k\big(\frac{1}{z}\big)\Big)\Big(\sum\hbar^{\ell}\beta_{\ell}\big(\frac{1}{z'}\big)\Big)\bigg) \Big/ \big(\frac{1}{z}+\frac{1}{z'}\big)\bigg.}.
\eea
Since we are in the hyperelliptic case, we have that $\beta_k\big(\frac{1}{z}\big)=\alpha_k\big(\frac{-1}{z}\big)$ and $\tilde{\beta}_k\big(\frac{1}{z}\big)=\tilde{\alpha}_k\big(\frac{-1}{z}\big)$, which implies that the first factor of the denominator is divisible by $\frac{1}{z}+\frac{1}{z'}$. Hence the series at the denominator is invertible.

Therefore, at each order of $\hbar$, the coefficient is a polynomial of $1/z,1/z'$, which tends to $0$ at $z,z'\to\infty$:
$$
H(z,z')-1 \in \frac{1}{zz'} \mathbb C[z^{-1},z'^{-1}][[\hbar]].
$$
We shall prove that $H=1$ by following the method of \cite{bergre2009determinantal}. Let us assume that $H\neq 1$, and write
$$
H(z,z') = 1+H_M(z) z'^{-M} + O(z'^{-M-1}),
$$
where $M\geq 1$ is the smallest possible power of $z'$ whose coefficient $H_M$ would be $\neq 0$ as a formal series of $\hbar$. Let $\tilde{K}(x,x')\coloneqq A(x)\tilde B(x')-\tilde A(x)B(x')$. Using equation~\eqref{eqKD} for $\tilde{\psi}(z,z')=H(z,z')\tilde{K}(x,x')e^F$, together with the fact that $\overline{K}(x,x')\coloneqq\tilde{K}(x,x')e^F$ satisfies the same equation~\eqref{eqKD} (which is proved in Appendix~\ref{appendix} in general), we obtain the following equation for $H=H(z,z')$:
\begin{align}\label{eqDHreduced}
&\frac{\dd^2 H}{\dd x'^2} +  2\frac{\dd\ln B(x')}{\dd x'}\frac{\dd H}{\dd x'} +   2\frac{\dd\ln (\tilde A(x)/A(x)-\tilde B(x')/B(x'))}{\dd x'}\frac{\dd H}{\dd x'} \nonumber \\
& \quad\quad -\tilde{K}(x,x')(H L(x').\tilde{K}(x,x')-H\tilde{K}(x,x')L(x'). F)=
 \frac{1}{x'-x} \left( \frac{\dd}{\dd x}+\frac{\dd}{\dd x'}\right)H,
\end{align}
whose leading power of $z'$ comes only from the second term and is
\beq
2y(z')H_M(z)z'^{-M-2}+O(z'^{-M-2})=2H_M(z)z'^{-M-1}+O(z'^{-M-2})=0,
\eeq
implying that $H_M=0$, and thus contradicting the hypothesis that $H\neq 1$. 
This proves the theorem.
\end{proof}

As a corollary this implies that
\beq
\lim_{z'\to\infty} \frac{(x(z)-x(z'))\psi([z]-[z'],t,\hbar)}{\tilde B(x(z'))} = A(x(z)).
\eeq
In other words
\beq
\frac{i}{2\sqrt{z}} e^{\hbar^{-1}\int_0^z \omega_{0,1}} \ e^{\sum_{(g,n)\neq (0,1),(0,2)} \frac{\hbar^{2g-2+n}}{n!} \int_{\infty}^z\dots\int_{\infty}^z \omega_{g,n}} = A(x(z)).
\eeq
In the Painlev\'e system, the function $A(x)$ is annihilated by the quantum curve
\beq
\hat y^2 -\left((x-U)^2(x+2U)+\frac{\hbar^2}{2}\ddot U (x-U)+\frac{\hbar^2}{4}\dot U^2\right) + \frac{\hbar^2}{2(x-U)} \dot U - \frac{\hbar}{x-U} \hat y.
\eeq

\subsection{General genus 0 case}

The same argument applies to any isomonodromic system thanks to the fact that both the $2$-point wave function $\tilde{\psi}=(x-x')\psi e^{F}$ constructed from topological recursion and the integrable kernel $\overline{K}(x,x')=\tilde{K}(x,x')e^F$ satisfy the same PDE \eqref{eqKD}. The first claim is the main result of the article and the second claim was proved in Appendix~\ref{appendix} in general, since we could not find this result in the literature for a general isomonodromic system.
Consider that we have a Lax system of the type
\beq\label{isoSystem_g0}
\left\{\begin{array}{rcl}
\hbar \frac{\partial}{\partial x}\Psi(x) &=& \mathcal L(x;\hbar)\Psi(x) , \\[0.2em]
\hbar \frac{\partial}{\partial t_k}\Psi(x) &=& \mathcal R_k(x;\hbar)\Psi(x).
\end{array}\right.
\eeq
whose spectral curve in the limit $\hbar\to 0$ is a genus zero curve of the form $\det(y-\mathcal L_0(x))=y^2-R(x)=0$, and which has a WKB formal power series solution of the form
\bea
A(x) &\sim & \frac{i}{\sqrt{2\dd x/\dd z}}\, e^{\hbar^{-1}\int_0^{x} y\dd x} (1 + O(\hbar)), \nonumber \\
B(x) &\sim & \frac{i}{\sqrt{2\dd x/\dd z}}\, e^{-\hbar^{-1}\int_0^{x} y\dd x} (1 + O(\hbar)), \nonumber \\
\tilde A(x) &\sim & i\sqrt{\frac12 \dd x/\dd z}\, e^{\hbar^{-1}\int_0^{x} y\dd x} (1 + O(\hbar) ),\nonumber \\
\tilde B(x) &\sim & -i\sqrt{\frac12 \dd x/\dd z}\, e^{-\hbar^{-1}\int_0^{x} y\dd x}  (1+ O(\hbar)), \nonumber
\eea
where each coefficient of higher powers of $\hbar$ in $(1+O(\hbar))$ is a rational function of $z$ that tends to 0 at poles $z\to \zeta$, where we have chosen a pole $x(\zeta)=\infty$. This pole has degree $-d=1$ or $-d=2$, and $\xi=x^{1/d}$ is a local coordinate near the pole.
We emphasize that for all genus 0 spectral curves where $y^2=P_{\text{odd}}(x)$ with $P_{\text{odd}}(x)$ an odd polynomial of $x$, such systems are explicitely known as Gelfand--Dikii systems \cite{Eynbook} described in Section~\ref{sec:GD} below, and for more general spectral curves (genus $0$, possibly with non-empty Newton polygon), it was proved in \cite{MarchalOrantin} that a $2\times 2$ autonomous system $\mathcal L_0$ always admits an $\hbar$-deformation $\mathcal L$ with this property.

We define the following formal series of $\hbar$:
$$
H(z,z') \coloneq \frac{(x(z)-x(z'))\psi([z]-[z'])}{A(x)\tilde B(x')-\tilde A(x)B(x')} = 1+O(\hbar).
$$
We have proved that the integrable kernel $\overline{K}(x,x')$ associated to the Lax system satisfies equation \eqref{eqKD} (see Theorem~\ref{thmApp}), and thus 
$H$ satisfies  the PDE \eqref{eqDHreduced}.

Moreover the coefficients of $H$ are analytic functions of $z'$, which tend to $0$ at $z'\to\zeta$. Let us write 
$H(z,z')=1+O(x'^{1/d})$.
The subleading coefficient of $H=1+H_M(z) x'^{M/d}+O(x'^{M/d-1})$ must satisfy $y(z')H_M(z)x'^{M/d-1}=O(x'^{M/d-2})$, and therefore $H=1$.

This implies that $\psi([z]-[z'])$ coincides with the integrable kernel
$$
\psi([z]-[z']) = \frac{A(x(z))\tilde B(x(z'))-\tilde A(x(z))B(x(z'))}{x(z)-x(z')}.
$$

Then taking the limit $z'\to\zeta$ this implies that
$$
\lim_{z'\to\zeta} \frac{(x(z)-x(z'))\psi([z]-[z'])}{\tilde B(x(z'))} = A(x(z)).
$$
Knowing that $A(x),\tilde A(x)$ satisfy an isomonodromic system with first equation
\beq
\hbar \frac{\partial}{\partial x} \begin{pmatrix}A(x) \cr \tilde A(x)\end{pmatrix} =\mathcal L(x,t,\hbar) \begin{pmatrix}A(x) \cr \tilde A(x)\end{pmatrix},
\eeq
with
\beq
\mathcal L(x,t,\hbar) = \begin{pmatrix}
\alpha(x,t,\hbar) & \beta(x,t,\hbar) \cr
\gamma(x,t,\hbar) & \delta(x,t,\hbar)
\end{pmatrix},
\eeq
where $\alpha,\beta,\gamma,\delta$ are rational functions of $x$, with coefficients being formal power series of $\hbar$, 
we get the quantum curve annihilating $A(x)$:
\beq
\hat y^2 -(\alpha+\delta)\hat y +(\alpha\delta-\beta\gamma) - \hbar\left( \dd\alpha/\dd x - \alpha\frac{\dd \beta/\dd x}{\beta} + \frac{\dd \beta/\dd x}{\beta}\hat y\right).
\eeq
Its classical part $\hbar\to 0$ is indeed the spectral curve
\beq
y^2-(\alpha(x,t,0)+\delta(x,t,0))y (\alpha(x,t,0)\delta(x,t,0)-\beta(x,t,0)\gamma(x,t,0)) = \det (y\mathrm{Id}-\mathcal L_0(x)).
\eeq

\subsection{Higher genus case}

If the curve $y^2=R(x)$ has genus $\hat{g}>0$, it was verified in \cite{Bouchard_Quantizing} (for $\hat{g}=1$), and argued in \cite{BEInt}, that the perturbative wave  function cannot satisfy the quantum curve: in fact just because it is not a function  (order by order in $\hbar$) on  the spectral curve. Indeed, multiple integrals of type $\int_o^z\dots\int_o^z \omega_{g,n}$ are not invariant after $z$ goes around a cycle, and do not transform as Abelian differentials.
It was argued in \cite{EO07, BEInt, GIS} that only the non-perturbative wave function of \cite{Eynard_2009, EMhol} can be a wave function and can obey a quantum curve, and this was proved up to the 3rd non trivial powers of $\hbar$ for arbitrary curves in \cite{BEInt}, and verified to many orders for elliptic curves in~\cite{Bouchard_Quantizing}. 

It is also useful to introduce the partition function, which is independent of any divisor: $Z(\hbar)=\psi(D=\emptyset,\hbar)$, namely
$$
Z(\hbar)\coloneqq \exp \Big(\sum_{g\geq 0}\hbar^{2g-2}\omega_{g,0}\Big),
$$
where the $\omega_{g,0}$ were also denoted by $F_g$. From now on, we omit the dependence on $\hbar$ of both partition and wave functions and we will introduce the dependence on the filling fractions.

Consider a curve $y^2=R(x)$ with genus $\hat{g}>0$ and let $\mathcal A_i\cap \mathcal B_j=\delta_{i,j}$ be a symplectic basis of cycles of $H_1(\Sigma,\mathbb Z)$ (i.e.~integer cycles). We choose the Bergman kernel normalized on $\mathcal A$-cycles.

Recall the 1st kind times 
\beq
\epsilon_j = \frac{1}{2\pi i} \oint_{\mathcal A_j} y\dd x, \;\; j=1,\ldots, \hat{g},
\eeq
and that the $\mathcal B_j$-period of $\omega_{g,n+1}$ is the variation of $\omega_{g,n}$ with respect to $\epsilon_j$:
\beq
\frac{\partial}{\partial\epsilon_j}\omega_{g,n}(z_1,\ldots,z_n) =\oint_{\mathcal B_j} \omega_{g,n+1}(\cdot,z_1,\ldots,z_n).
\eeq

Since \eqref{eqPDEgeneralL} is a linear PDE, any linear combination of solutions is a solution.
Moreover, since the coefficients of the PDE do not involve the times $\epsilon_j$, we remark that shifting $\epsilon_j\to \epsilon_j+n_j$ is another solution, which we denote as follows
\beq
\psi((\epsilon_j\to\epsilon_j+n_j)_j;[z]-[z']).
\eeq
The transseries linear combination introduced in \cite{Eynard_2009, EMhol, BEInt, GIS}
\beq\label{trans}
\hat\psi([z]-[z']) \coloneqq \frac{1}{\hat {\mathcal T}} \sum_{n_1,\dots,n_{\hat{g}}\in \mathbb Z^{\hat{g}}} 
\psi((\epsilon_j\to\epsilon_j+n_j)_j;[z]-[z']) Z((\epsilon_j\to\epsilon_j+n_j)_j),
\eeq
where
\beq
\hat{\mathcal T}\coloneqq\sum_{n_1,\dots,n_{\hat{g}}\in \mathbb Z^{\hat{g}}} 
 Z((\epsilon_j\to\epsilon_j+n_j)_j),
\eeq
is thus also a solution of the same PDE. The construction of the transseries \eqref{trans} is actually a sum over integer 1st kind $\mathcal{B}$-cycles, since
$$
\psi((\epsilon_j);[z+\mathcal{B}_j]-[z'])=e^{\frac{\partial}{\partial\epsilon_j}}\psi((\epsilon_j)_j;[z]-[z']) = \psi((\epsilon_1,\ldots,\epsilon_j\to\epsilon_j+1,\ldots,\epsilon_{\hat{g}});[z]-[z']).
$$

The combination $\hat\psi$ hence remains invariant \cite{GIS} if we modify the integration homotopy class from $z'$ to $z$ by adding 1st kind $\mathcal{B}$-cycles, and it is, order by order as a transseries of $\hbar$, a function of $z$ and $z'$ on the spectral curve. Shifts by other kinds of cycles of the spectral curve were already trivial for $\psi$, in the sense that they result in the multiplication by simple factors, which one expects from a solution of an isomonodromic system.

From there, the same argument as for the genus zero case applies.
Assume that there is an isomonodromic system 
\beq\label{isoSystem_anyg}
\left\{\begin{array}{rcl}
\hbar \frac{\partial}{\partial x}\Psi(x) &=& \mathcal L(x;\hbar)\Psi(x) , \\[0.2em]
\hbar \frac{\partial}{\partial t_k}\Psi(x) &=& \mathcal R_k(x;\hbar)\Psi(x).
\end{array}\right.
\eeq 
whose associated spectral curve is our spectral curve, with
$$
\Psi(x) = \begin{pmatrix}
A(x) & B(x) \cr
\tilde A(x) & \tilde B(x)
\end{pmatrix}
$$
a formal transseries solution.
Then  the formal transseries
\beq
H(z,z') = \frac{(x(z)-x(z'))\hat{\psi}([z]-[z'])}{A(x)\tilde B(x')-\tilde A(x)B(x')}
=1+O(\hbar)
\eeq
satisfies  the PDE \eqref{eqDHreduced}, and is such that $H(z,z')=1+O(x'^{1/d})$.
The subleading order of $H=1+H_M(z) x'^{M/d}+O(x'^{M/d-1})$ must satisfy $y(z')H_M(z)x'^{M/d-1}=O(x'^{M/d-2})$, and therefore $H=1$.
This implies that
\beq
\hat{\psi}([z]-[z']) = \frac{A(x(z))\tilde B(x(z'))-\tilde A(x(z))B(x(z'))}{x(z)-x(z')}.
\eeq
This also implies that
\beq
\lim_{z'\to\zeta} \frac{(x(z)-x(z'))\hat{\psi}([z]-[z'])}{\tilde B(x(z'))} = A(x(z)).
\eeq
Since $A(x),\tilde A(x)$ satisfy the isomonodromic system
$$
\hbar \frac{\partial}{\partial  x} \begin{pmatrix}A(x) \cr \tilde A(x)\end{pmatrix} =\mathcal L(x) \begin{pmatrix}A(x) \cr \tilde A(x)\end{pmatrix},
$$
where
$$
\mathcal L(x,t,\hbar) = \begin{pmatrix}
\alpha(x,t,\hbar) & \beta(x,t,\hbar) \cr
\gamma(x,t,\hbar) & \delta(x,t,\hbar)
\end{pmatrix},
$$
we find the quantum curve annihiliating $A(x)$:
\beq
\hat y^2 -(\alpha(x)+\delta(x))\hat y +(\alpha(x)\delta(x)-\beta(x)\gamma(x)) - \hbar\left( \alpha'(x) - \alpha(x)\frac{\beta'(x)}{\beta(x)} + \frac{\beta'(x)}{\beta(x)}\hat y\right).
\eeq
Its classical part $\hbar\to 0$ is indeed the equation
\beq
\det (y\mathrm{Id}-\mathcal L(x,t,0))=0.
\eeq

\subsection{Examples: Gelfand--Dikii systems}
\label{sec:GD}
These systems generalize the Painlev\'e I equation; they appear in the enumeration of maps in the large size limit \cite{Eynbook}.
For these Gelfand--Dikii systems, the proof that $\psi([z]-[z'])$ coincides with the integrable kernel (which then implies the quantum curve) can be found in \cite[Chapter 5]{Eynbook}, by another method. Here let us provide another proof with our current method.

The Gelfand--Dikii polynomials are defined as differential polynomials of a function $U(t)$, by the recursion
\beq
R_0(U) = 2  \,, \quad\quad
\frac{\d }{\d t} R_{k+1}(U) = -2U \frac{\d R_k(U)}{\d t} - R_k(U) \frac{\d U}{\d t} + \frac{\hbar^2}{4} \frac{\d ^2 R_k(U)}{\d t^2}.
\eeq
At each step the integration constant is chosen so that $R_k(U)$ is homogeneous in powers of $U$ and $\d ^2/\d t^2$.
The first few are given by
\bea
R_0 & = & 2, \nonumber \\
R_1 & = & -2U, \nonumber \\
R_2 & = & 3 U^2 - \frac{\hbar^2}{2}\ddot U, \nonumber \\
R_3 & = & -5 U^3 + \frac{5\hbar^2}{2} U \ddot U+\frac{5\hbar^2}{4}\dot U^2 - \frac{\hbar^4}{8}\frac{\d ^4U}{\d t^4}.
\eea
Let $m\geq 1$ be an integer, and let $\tilde t_0,\tilde t_1,\tilde t_2,\dots,\tilde t_m$ be a set of ``times''.
Let $U(t;\tilde t_0,\tilde t_1,\tilde t_2,\dots,\tilde t_m)$ be a solution of the following non-linear ODE:
\beq\label{GelfandDikkiiUODE}
\sum_{j=0}^m \tilde t_j R_{j+1}(U) = t.
\eeq
Notice that, formally, $t=-2\tilde t_{-1}$.
For $m=1$, the equation \eqref{GelfandDikkiiUODE} is the Painlev\'e I equation. The case $m=2$ is called the Lee--Yang equation. The case $m=0$ is simply $U(t)=-\frac{t}{2\tilde t_0}$.

Consider the Lax pair (adopting the normalizations of \cite{Eynbook}) given by
\beq
\mathcal R(x,t,\hbar) = \begin{pmatrix}
0 & 1 \cr
x+2U(t) & 0
\end{pmatrix}
\eeq
and
\beq
\mathcal L(x,t,\hbar) = \sum_{j=0}^{m} \tilde t_j \mathcal{L}_j(x,t,\hbar),
\eeq
where
$$
\mathcal L_j(x,t,\hbar)
=\begin{pmatrix}
\alpha_j(x,t) & \beta_j(x,t) \cr
\gamma_j(x,t) & -\alpha_j(x,t)
\end{pmatrix} \text{ with},
$$
$$
\beta_j(x,t) = \frac12 \sum_{k=0}^j x^{j-k}R_k(U)
 , \quad
\alpha_j(x,t) = -\frac{\hbar}{2} \frac{\d}{\d t}\beta_j(x,t)
 , \quad
\gamma_j(x,t) = (x+2U)\beta_j(x,t)+\hbar \frac{\d}{\d t}\alpha_j(x,t).
$$
One can easily verify that the Gelfand--Dikii polynomials are such that the zero curvature equation is satisfied
\beq
\hbar \frac{\d}{\d t} \mathcal L(x,t,\hbar) + \hbar \frac{\d}{\d x}\mathcal R(x,t,\hbar) = [\mathcal R(x,t,\hbar),\mathcal L(x,t,\hbar)].
\eeq

The differential equation \eqref{GelfandDikkiiUODE} admits a formal power series solution $U(t,\hbar)$ with only even powers of $\hbar$:
\beq
U(t,\hbar) = \sum_k \hbar^{2k} u_k(t),
\eeq
whose first term $u(t)=u_0(t)$ satisfies an algebraic equation
$$
\sum_{j=0}^m \frac{(2j+1)!}{j!(j+1)!} \tilde t_j (-u/2)^{j+1} = -\frac14 t.
$$
In the Painlev\'e I case, $m=1$, $\tilde t_k=\delta_{k,1}$, we recover $t=-3u^2$.

The spectral curve, in the limit $\hbar\to 0$:
\beq
\det (y-\mathcal L(x,t,0)) =0
\eeq
is always a genus 0 curve. It admits the rational parametrization
\beq
\left\{\begin{array}{l}
x(z) = z^2-2u(t) \cr
y(z) = \sum_{j=0}^m \tilde t_j \left(z^{2j+1}(1-2u(t)/z^2)^{j+\frac12}\right)_+
\end{array},\right.
\eeq
where $()_+$ means the positive part in the Laurent series expansion near $z=\infty$, i.e.
$$
y(z) = \sum_{j=0}^m \tilde t_j \sum_{k=0}^{j} (-u)^k \frac{(2j+1)!!}{(2j-2k+1)!!} \ z^{2j-2k+1}.
$$
In the Painlev\'e I case, $\tilde t_j=\delta_{j,1}$, we recover $y(z)=z^3-3uz$.\\
In the case $m=0$, with $\tilde t_0=1$, we recover the Airy system
\beq
\mathcal L(x,t,\hbar) = \begin{pmatrix}
0 & 1 \cr
x-t & 0
\end{pmatrix}
\eeq
with spectral curve $y^2=x-t$.

Let $\Psi(x,t,\hbar)$ as follows
$$
\Psi(x,t,\hbar) = \begin{pmatrix}
A(x) & B(x) \cr
\tilde A(x) & \tilde B(x)
\end{pmatrix}
$$
be a WKB  $\hbar$ formal series solution of
\beq
\hbar\frac{\d}{\d x} \Psi(x,t,\hbar) = -\mathcal L(x,t,\hbar)\Psi(x,t,\hbar)
\, , \quad\quad
\hbar\frac{\d}{\d t} \Psi(x,t,\hbar) = \mathcal R(x,t,\hbar)\Psi(x,t,\hbar).
\eeq
Our previous results show that the formal series $\psi([z]-[\infty],\tilde t,\hbar)$ coincide with 
\beq
A(x,t,\hbar) = \frac{1}{\sqrt{2z}}e^{\hbar^{-1}\int_0^z ydx} e^{\sum_{(g,n)\neq (0,1),(0,2)} \frac{\hbar^{2g-2+n}}{n!} \int_{\infty}^z\dots\int_\infty^z \omega_{g,n}},
\eeq
and is annihilated by the quantum curve
\beq
\hat y^2 -(\alpha(x)+\delta(x))\hat y +(\alpha(x)\delta(x)-\beta(x)\gamma(x)) - \hbar\left( \alpha'(x) - \alpha(x)\frac{\beta'(x)}{\beta(x)} + \frac{\beta'(x)}{\beta(x)}\hat y\right).
\eeq

\appendix

\section{PDE for any isomonodromic system}\label{appendix}

The goal of this appendix is to show that the integrable kernel $\overline{K}(x,x')\coloneqq\tilde{K}(x,x')e^F$ of an isomonodromic system satisfies the same PDE \eqref{eqKD} that we obtained for the $2$-point wave function built from topological recursion. Here $F=\log \mathcal T$, with $\mathcal T$ the tau function of the isomonodromic system.

Consider a solution of a $2\times 2$ isomonodromic system written as:
$$
\Psi(x) = \begin{pmatrix}
A(x) & B(x) \cr
\tilde A(x) & \tilde B(x)
\end{pmatrix},
\quad
\det \Psi(x)=1,
$$
i.e.~$\Psi(x)$ satisfying the (compatible) system of equations:
\beq\label{isoSystem}
\left\{\begin{array}{rcl}
\hbar \frac{\partial}{\partial x}\Psi &=& \mathcal L(x;\hbar)\Psi , \\[0.2em]
\hbar \frac{\partial}{\partial t_k}\Psi &=& \mathcal R_k(x;\hbar)\Psi.
\end{array}\right.
\eeq
The equations with respect to the isomonodromic times $t_k$ can be seen as isomonodromic deformations of the first equation.

Consider the deformed spectral curve
\beq
P(x,y;\hbar)=\det(y\,\mathrm{Id}-\mathcal L(x;\hbar))=y^2+R(x)+\sum_{m\geq 1}\hbar^mP_m(x,y)=P_0(x,y)+O(\hbar).
\eeq
We will make use of the following technical assumption: 
\begin{assumption}\label{assumption}
Let $\mathcal{N}$ be the Newton polygon associated to $P_0(x,y)=y^2-R(x)$ and let $\mathring{\mathcal{N}}$ be its interior. For $m\geq 1$ we only allow $P_m(x,y)=\sum_{i,j}P_{i,j}x^i y^j$ whose only non-zero coefficients $P_{i,j}\neq 0$ are such that $(i+1,j+1)\in\mathring{\mathcal{N}}$, which in our case only allows for $j=1$, so $P_m(x,y)=P_m(x)$. 
\end{assumption}

\begin{remark}
This assumption is to ensure that $P$ has the same Casimirs as $P_0$. It is always possible to transform a deformed spectral curve into one satisfying this assumption, by redefining what we mean by Casimirs of $P_0$, but here we assume we already have this shape for simplicity. 
\end{remark}

The associated classical spectral curve $\Sigma = \{(x,y)\mid y^2=R(x)\}$ is presented as a two-sheeted covering of the Riemann sphere $x:\Sigma \rightarrow \mathbb{P}^1$.

We assume that the entries of the matrix $\mathcal{L}$ are rational functions of $x$:
$$
\mathcal{L}(x;\hbar)=\sum_{l=1}^N\sum_{j=0}^{m_l}\frac{\mathcal{L}^{(l)}_j}{(x-\lambda_l)^{j+1}}-\sum_{j=1}^{m_{\infty}}\mathcal{L}^{(\infty)}_{j-1}x^{j-1}.
$$
We often omit the dependence on the variables $\hbar$ and even $x$ for simplicity.
The solutions $\Psi(x)$ of the first linear differential equation of \eqref{isoSystem} have essential singularities at $x=\lambda_l$ and $x=\infty$. Let us define $\sigma_3\coloneqq\left(\begin{array}{lr} 1 & 0 \\ 0 & -1\end{array}\right)$. We know that around each pole of $\mathcal{L}(x)$ the function $\Psi(x)$ admits an asymptotic expansion of the form
\beq 
\hat{\Psi}(x) e^{\hbar^{-1}\sigma_3T(x)},
\eeq
where $\hat{\Psi}(x)$ is regular at the poles $P\in \{\infty, \lambda_1,\ldots,\lambda_l\}$ of $\mathcal{L}(x)$ and 
\begin{align*}
T(x)\underset{x\rightarrow P}{\sim} \tilde{t}_{P,0}\log \xi_P + \sum_{j=1}^{m_P} \frac{\tilde{t}_{P,j}}{\xi_P^j},
\end{align*}
where $\xi_P=x^{d_{\infty}}$, with $d_{\infty}$ equal to $-1$ (unramified) or $-2$ (ramified), for $P=\infty$, and $\xi_P=x-\lambda_l$, for $P=\lambda_l$.

The isomonodromic deformation parameters $t_k=T_{\alpha,j}$ of \eqref{isoSystem}, using the multi-index notation $k=(\alpha,j)$, include:
\begin{itemize}
\item $\alpha = \infty$: $T_{\infty,j}\coloneqq\tilde{t}_{\infty,j}=t_{\zeta_{\infty}^+,j}$, for $j=1,\ldots,m_{\infty}$.
\begin{itemize}
\item If $\infty$ is an unramified critical value of $x$, i.e.~$d_{\infty}=-1$ and $x^{-1}(\infty)=\{\zeta_{\infty}^+,\zeta_{\infty}^-\}$, then $t_{\zeta_{\infty}^+,j}=-t_{\zeta_{\infty}^-,j}$.
\item If $\infty$ is a ramified critical value of $x$, i.e.~$d_{\infty}=-2$ and $\zeta_{\infty}^+=\zeta_{\infty}^-$.
\end{itemize}
\item $\alpha = l=1,\ldots,N$: 
\begin{itemize}
\item For $j=1,\ldots,m_{l}$, $T_{l,j}\coloneqq\tilde{t}_{\lambda_l,j}=t_{\zeta_{l}^+,j}=-t_{\zeta_{l}^-,j}$. 
\item For $j=-1$, $T_{l,-1}\coloneqq \lambda_l$.
\end{itemize}
\end{itemize} 

We will also use this multi-index notation for the matrices $\mathcal{R}_k=\mathcal{R}_{\alpha,j}$ of the system~\eqref{isoSystem}.

\begin{remark}
From our Assumption~\ref{assumption}, we have that all $t_i$ are independent of $\hbar$ and coincide with the moduli of the classical spectral curve $y^2=R(x)$.
\end{remark}

\begin{remark}
Since $\mathrm{Tr}\,\mathcal{L} =0$, we have 
\beq\label{eqrmk}
\mathcal{L}^2=-\det\mathcal{L} \cdot \mathrm{Id}.
\eeq
\end{remark}

Let $\mathbb{K}(x,x')\coloneqq \Psi^{-1}(x')\Psi(x)$. Observe that $\tilde{K}(x,x')=A(x)\tilde B(x')-\tilde A(x)B(x')=\mathrm{Tr} \left(\mathbb{K}\begin{pmatrix}1 & 0 \\ 0 & 0\end{pmatrix}\right)$. 
We called $\zeta_i$ the poles of $y\dd x$. Recall the operator defined by:
$$
\mathcal D \coloneqq \hbar^2\frac{\d^2}{\d x^2}- \hbar^2L(x)-R(x),
$$
where
\bea
L(x) &\coloneqq &
 \sum_{i, x(\zeta_i)=\infty}  \sum_{j=1-2d_i}^{m_i} t_{\zeta_i,j} \sum_{0\leq k \leq \frac{1-j}{d_i}-2} x^k  \Big(-\frac{j}{d_i}-k-2\Big) \partial_{\mathcal{B}_{\zeta_i,j+d_i(k+2)}}   \\
&& +  \sum_{i, x(\zeta_i)\neq \infty} \sum_{j=0}^{m_i} t_{\zeta_i,j} \sum_{k=0}^{j} (x-x(\zeta_i))^{-(k+1)} (j+1-k) \partial_{\mathcal{B}_{\zeta_i,j+1-k}}, \nonumber
\eea
with 
$$
\d_{\mathcal{B}_{p,k}}\omega_{g,n}(z_1,\ldots,z_n)\coloneqq \int_{\mathcal{B}_{p,k}}\omega_{g,n+1}(\cdot,z_1,\ldots,z_n) =\Res_{x\to p} \frac{\xi_p^k}{k}\omega_{g,n+1}(\cdot,z_1,\ldots,z_n),
$$
for every second type cycle $\mathcal{B}_{p,k}$, $k\geq 1$, which we introduced in \eqref{2ndTypeCycles}.

From Corollary~\ref{isomonodromic}, the operator $L(x)$ can be re-written in terms of derivatives with respect to the isomonodromic parameters: $t_{\zeta_i,k}$, $k=1,\ldots,m_i$, and $\lambda_l$, for $l=1,\ldots,N$.

\begin{theorem}\label{thmApp}
The operator 
\beq 
\mathcal{O}\coloneqq \mathcal D - \hbar^2 L(x).F - \frac{\hbar^2}{x-x'} \left(\frac{\d}{\d x}+\frac{\d}{\d x'}\right)
\eeq
annihilates the integrable kernel:
\beq\label{eqKDapp}
\mathcal O\, \mathbb{K}(x,x') = 0.
\eeq

\end{theorem}
The differential equation \eqref{eqKDapp} is equivalent to  
$$
\left(\mathcal{D}- \frac{\hbar^2}{x-x'} \left(\frac{\d}{\d x}+\frac{\d}{\d x'}\right)\right)(\mathbb{K}(x,x') e^F)=0.
$$

The proof of this theorem will follow from the three following lemmas.

\begin{lemma}
Let $\mathcal T$ be the tau function of the isomonodromic system and $F\coloneqq\log \mathcal T$. Then,
\beq\label{lemmaMJ} 
\frac{1}{2}\mathrm{Tr}\,\mathcal{L}(x;\hbar)^2=-\det \mathcal{L}(x;\hbar)=\hbar^{-2}R(x)+L.F.
\eeq
\end{lemma}
\begin{proof}
It was proved by Jimbo--Miwa--Ueno \cite{MJU81} that there exists a tau function $\mathcal T(\{t_k\})$ of the isomonodromic times $t_k=T_{\alpha,j}$ such that
\beq\label{MJeq}
\frac{\d}{\d t_k} \log \mathcal T =-\underset{x=P}{\mathrm{Res}}\; \mathrm{Tr} \left(\frac{\d}{\d t_k} T(x) \sigma_3 \Psi^{-1}(x)\frac{\partial}{\partial x}\Psi(x)\right), 
\eeq
where $P=\infty$, if $\alpha=\infty$, and $P=\lambda_l$, if $\alpha=l$, for $l=1,\ldots,N$. Let $y_s$ be the singular part of $y$. We have
$$
\frac{\d}{\d x}\Psi =\frac{\d}{\d x}\hat{\Psi}e^{\hbar^{-1}\sigma_3 T}+\hbar^{-1}\hat{\Psi}\sigma_3 y_s e^{\hbar^{-1}\sigma_3 T},
$$
which implies 
$$
\mathcal{L}=\frac{\d}{\d x}\Psi\cdot\Psi^{-1}=\frac{\d}{\d x}\hat{\Psi}\cdot\hat{\Psi}^{-1}+\hbar^{-1}\hat{\Psi}\sigma_3 y_s\hat{\Psi}^{-1}.
$$
Therefore
\beq\label{Leq}
\mathrm{Tr}\, \mathcal{L}^2 = \mathrm{Tr}\, \bigg(\frac{\d}{\d x}\hat{\Psi}\cdot\hat{\Psi}^{-1}\bigg)^2+2\hbar^{-2}y_s^2+2\mathrm{Tr}\, \bigg(y_s\sigma_3\hat{\Psi}^{-1}\frac{\d}{\d x}\hat{\Psi}\bigg).
\eeq
For every $P\in\{\infty,\lambda_1,\ldots,\lambda_l\}$, $\zeta_i\in x^{-1}(P)$, we have
\beq\label{beh1}
\frac{\d}{\d t_{\zeta_i,k}}T(x)=-\frac{\xi_P^{-k}}{k},
\eeq
for $k=1,\ldots,m_i$. Around $\zeta_l\in x^{-1}(\lambda_l)$, we have
\beq\label{beh2}
\frac{\d}{\d \lambda_l}T(x)= -\sum_{j=0}^{m_l}-t_{\zeta_l,j}\xi_{\lambda_l}^{-j-1},
\eeq
which behaves as $-y(z)$ around $z=\zeta_l$.
Substituting \eqref{Leq} in \eqref{MJeq}, we get
$$
\frac{\d}{\d t_k} F=-\underset{x=P}{\mathrm{Res}}\;\frac{1}{2y_s(x)}\frac{\d}{\d t_k} T(x) \bigg(\mathrm{Tr}\, \mathcal{L}(x)^2 -\mathrm{Tr}\, \bigg(\frac{\d}{\d x}\hat{\Psi}(x)\cdot\hat{\Psi}(x)^{-1}\bigg)^2-2\hbar^{-2}y_s(x)^2 \bigg).
$$
Since $\hat{\Psi}(x)$ is analytic at $x=P$, we have $\frac{\d}{\d x}\hat{\Psi}(x)\cdot\hat{\Psi}(x)^{-1}=O(\xi_P^2)$ for $x\rightarrow P$. Thus the middle term does not contribute to the residue:
$$
\frac{\d}{\d t_k} F=-\underset{x=P}{\mathrm{Res}}\;\frac{1}{2y_s(x)}\frac{\d}{\d t_k} T(x) \big(\mathrm{Tr}\, \mathcal{L}(x)^2 -2\hbar^{-2} R(x) \big).
$$
Using the behaviors \eqref{beh1} and \eqref{beh2}, we get that $\mathrm{Tr}\, \mathcal{L}(x)^2 = 2(L(x).F+\hbar^{-2}R(x))$.
\end{proof}

\begin{lemma}\label{lemmashape0}
Let $\zeta_{\infty}\in x^{-1}(\infty)$ and $\zeta_{l}\in x^{-1}(\lambda_l)$ be poles of $\omega_{0,1}$ of orders $m_{\infty}$ and $m_l$, $l=1,\ldots,N$, respectively. Let $d_{\infty}\coloneqq \mathrm{ord}_{\zeta_{\infty}}(x)$. We have
\beq\label{shape0}
\mathcal O\, \mathbb{K}(x,x') = \bigg(\Psi(x')^{-1}\hbar^2\left(\frac{\d \mathcal{L}(x;\hbar)}{\d x}-\frac{\mathcal{L}(x;\hbar)-\mathcal{L}(x';\hbar)}{x-x'}-\mathcal{R}(x,x')\right)\Psi(x')\bigg)\mathbb{K}(x,x'),
\eeq
with
\begin{align*}
\mathcal{R}(x,x') & \coloneqq \sum_{j=1-2d_{\infty}}^{m_{\infty}} t_{\zeta_{\infty},j} \sum_{s= 1}^{j+2d_{\infty}} \xi_{\infty}^{s-j-2d_{\infty}}  \Big(\frac{-s}{d_{\infty}}\Big) (\mathcal{R}_{\infty,s}(x)-\mathcal{R}_{\infty,s}(x')) \\
& + \sum_{l=1}^N \bigg(\xi_{\lambda_l}^{-1}(\mathcal{R}_{l,-1}(x)-\mathcal{R}_{l,-1}(x'))+ \sum_{s=1}^{m_l-1}\sum_{j=s}^{m_l-1} t_{\zeta_l,j}\,  \xi_{\lambda_l}^{-j+s-2} s (\mathcal{R}_{l,s}(x)-\mathcal{R}_{l,s}(x')) \bigg).
\end{align*}
\end{lemma}
\begin{proof}
On the one hand, we have
\begin{align}
\frac{\partial}{\partial x}\mathbb{K}(x,x') & = \Psi^{-1}(x')\frac{\partial}{\partial x}\Psi(x)=\Psi^{-1}(x')\mathcal L(x;\hbar)\Psi(x)= \Psi^{-1}(x')\mathcal L(x;\hbar)\Psi(x')\mathbb{K}(x,x');\label{eq1}\\
\frac{\partial}{\partial x'}\mathbb{K}(x,x') & = -\Psi^{-1}(x')\frac{\partial}{\partial x'}\Psi(x')\Psi^{-1}(x')\Psi(x)=-\Psi^{-1}(x')\mathcal L(x';\hbar)\Psi(x')\mathbb{K}(x,x'); \label{eq2}\\
\frac{\partial^2}{\partial x^2}\mathbb{K}(x,x') & = \Psi^{-1}(x')\left(\frac{\partial}{\partial x}\mathcal L(x;\hbar)+\mathcal L(x;\hbar)^2\right)\Psi(x')\mathbb{K}(x,x'). \label{eq3}
\end{align}
On the other hand, 
\beq\label{eqts}
\frac{\partial}{\partial t_k}\mathbb{K}(x,x') =  \Psi^{-1}(x')\frac{\partial}{\partial t_k}\Psi(x) + \frac{\partial}{\partial t_k}\Psi^{-1}(x')\Psi(x)=\Psi^{-1}(x')(\mathcal{R}_k(x;\hbar)-\mathcal{R}_k(x';\hbar))\Psi(x).
\eeq

Using \eqref{eq1}, \eqref{eq2}, \eqref{eq3} and \eqref{eqrmk}, we can rewrite
\begin{align*}
&(\mathcal D + L(x).F) \mathbb{K}(x,x') - \frac{\hbar^2}{x-x'} \left(\frac{\d}{\d x}+\frac{\d}{\d x'}\right)\mathbb{K}(x,x')=\\
&\bigg(\Psi(x')^{-1}\hbar^2\left(\frac{\d \mathcal{L}(x;\hbar)}{\d x}-\det \mathcal{L}(x;\hbar)-\frac{\mathcal{L}(x;\hbar)-\mathcal{L}(x';\hbar)}{x-x'}\right)\Psi(x') \\ &-(\hbar^2 L(x)+R(x)+\hbar^2L(x).F)\bigg)\mathbb{K}.
\end{align*}

Making use of \eqref{lemmaMJ}, we obtain
\beq\label{shape0}
\mathcal{O}\,\mathbb{K}=\bigg(\Psi(x')^{-1}\hbar^2\left(\frac{\d \mathcal{L}(x;\hbar)}{\d x}-\frac{\mathcal{L}(x;\hbar)-\mathcal{L}(x';\hbar)}{x-x'}\right)\Psi(x')-\hbar^2 L(x)\bigg)\mathbb{K}.
\eeq

Using the shape of the operator $L(x)$ from Corollary~\ref{isomonodromic} in terms of derivatives with respect to the isomonodromic parameters and the action on $\mathbb{K}$ by these derivatives given by \eqref{eqts}, we deduce 
$$
L(x)\mathbb{K}=\Psi^{-1}(x')\mathcal{R}(x,x')\Psi(x')\mathbb{K}.
$$
\end{proof}

Let us define $T(x) \coloneqq \int^x y(x)\dd x$ and $\sigma_3\coloneqq\left(\begin{array}{lr} 1 & 0 \\ 0 & -1\end{array}\right)$. We can write 
$$
\mathbb{K}(x,x') = e^{-\hbar^{-1}\sigma_3 T(x')} \hat{\mathbb{K}}(x,x') e^{\hbar^{-1}\sigma_3 T(x)}
$$
with $\hat{\mathbb{K}}(x,x')$ analytic when $x' \rightarrow \zeta_i$, where the $\zeta_i$ are the poles of $y\dd x$. Let us use the notation $\sigma\coloneqq \hbar^{-1}\sigma_3$.
\begin{lemma}\label{lemmashape1} 
The expression
\beq\label{eqshape1}
e^{\sigma T(x')} (\mathcal{O} \mathbb{K}(x,x'))\hat{\mathbb{K}}(x,x')^{-1}e^{-\sigma T(x')}
\eeq
is a rational function of $x$ and $x'$ with no poles at $x'\rightarrow\infty$ or $x'\rightarrow\lambda_l$.
\end{lemma}

\begin{proof}
The idea now is to rewrite the operator acting on $\hat{\mathbb{K}}=\hat{\mathbb{K}}(x,x')$, instead of on $\mathbb{K}=\mathbb{K}(x,x')$, and analyze the poles of the equation, as a function of $x'$. We have:
\begin{align*}
\frac{\d}{\d x}\mathbb{K}&=e^{-\sigma T(x')}\Big(\frac{\d}{\d x}\hat{\mathbb{K}}+\hat{\mathbb{K}}\,\sigma\, y(x)\Big)e^{\sigma T(x)}.\\
\frac{\d}{\d x}\mathbb{K} \cdot \mathbb{K}^{-1}&=e^{-\sigma T(x')}\Big(\frac{\d}{\d x}\hat{\mathbb{K}}\,\hat{\mathbb{K}}^{-1}+y(x)\hat{\mathbb{K}}\,\sigma\hat{\mathbb{K}}^{-1}\Big)e^{\sigma T(x')}.\\
\frac{\d}{\d x'}\mathbb{K} \cdot \mathbb{K}^{-1}&=e^{-\sigma T(x')}\Big(-y(x')\sigma+\frac{\d}{\d x'}\hat{\mathbb{K}}\,\hat{\mathbb{K}}^{-1}\Big)e^{\sigma T(x')}.
\end{align*}
\vspace{-0.2cm}
\begin{align*}
\frac{\d^2}{\d^2 x}\mathbb{K}\cdot \mathbb{K}^{-1}&=e^{-\sigma T(x')}\bigg(\frac{\d^2}{\d^2 x}\hat{\mathbb{K}}+2y(x) \frac{\d}{\d x}\hat{\mathbb{K}}\,\sigma +\frac{\d}{\d x}y(x)\hat{\mathbb{K}}\,\sigma+\hbar^{-2}y(x)^2\hat{\mathbb{K}} \bigg)\hat{\mathbb{K}}^{-1}e^{\sigma T(x')}.\\
L(x).\mathbb{K}\cdot \mathbb{K}^{-1}& =e^{-\sigma T(x')}\Big(-L(x).T(x')\sigma+L(x).T(x)\hat{\mathbb{K}}\sigma\hat{\mathbb{K}}^{-1}+L(x).\hat{\mathbb{K}} \hat{\mathbb{K}}^{-1}\Big)e^{\sigma T(x')}.
\end{align*}
Therefore we can write
\beq\label{lastExp}
e^{\sigma T(x')} (\mathcal{O} \mathbb{K}(x,x'))\hat{\mathbb{K}}(x,x')^{-1}e^{-\sigma T(x')}= \frac{y(x')\sigma}{x-x'}+L(x).T(x')\sigma+ F(x,x'),
\eeq
with $F(x,x')$ a rational function of $x$ and $x'$ with no poles at $x'\rightarrow\infty$ or $x'\rightarrow\lambda_l$. Let $z'$ be one of the preimages of $x'$: $x'=x(z')$.
\begin{itemize}
\item Behavior at $x'\rightarrow\infty$: We consider the behavior of $L(x).T(x')$ at $z'\rightarrow \zeta_{\infty}$, with $\zeta_{\infty}\in x^{-1}(\infty)$. Here we call $d=-d_{\infty}$ and $m=m_{\infty}$ to simplify the notations, and use $\xi=x^{-\frac{1}{d}}, \xi'=x'^{-\frac{1}{d}}$ as local coordinates. The only terms of $L_{\infty}(x)$ that may bring poles at $z'\rightarrow \zeta_{\infty}$ when applied to $T(x')$ read
$$
L_{\infty}(x)=  \sum_{j=1+2d}^{m} t_{\zeta_{\infty},j} \sum_{k= 0}^{\frac{j-1}{d}-2} \xi^{-dk}  \Big(\frac{j}{d}-k-2\Big) \frac{\partial}{\partial t_{\zeta_{\infty},j-d(k+2)}}
$$
We have
$$
\frac{s}{d}\frac{\d}{\d t_{\zeta_{\infty},s}}T(x')= -\frac{1}{d}\xi'^{-s}.
$$
Therefore, around $z'= \zeta_{\infty}$, we get
\beq\label{tocancel}
L_{\infty}(x).T(x')=-\frac{1}{d}\sum_{j=1+2d}^{m} t_{\zeta_{\infty},j} \sum_{k= 0}^{\frac{j-1}{d}-2} \xi^{-dk}\xi'^{-j+(k+2)d},
\eeq
which is indeed singular when $z'\rightarrow\zeta_{\infty}$. On the other hand, we consider the local behavior of the first term of the RHS of \eqref{lastExp} around $z'= \zeta_{\infty}$
$$
\frac{y(x')}{x-x'}=-\frac{y(x')}{x'-x}=\frac{1}{d}\sum_{j=0}^m t_{\zeta_{\infty},j}\sum_{k\geq 0}\xi'^{-j+(k+2)d}\xi^{-dk},
$$
whose singular terms around $z'= \zeta_{\infty}$ cancel exactly \eqref{tocancel}.

\item Behavior at $x'\rightarrow\lambda_l$: We consider the behavior of $L(x).T(x')$ at $z'\rightarrow \zeta_l$, with $\zeta_l\in x^{-1}(\lambda_l)$. The only terms of $L_{\Lambda}$ that may bring poles at $z'\rightarrow \zeta_l$ applied to $y(x')\dd x'$ read
\begin{align*}
&\sum_{s=1}^{m_l+1}\sum_{j=s-1}^{m_l} t_{\zeta_l,j}(x-\lambda_l)^{-j+s-2}s\, \d_{\mathcal{B}_{\zeta_l,s}}\omega_{0,1}(z') =\sum_{s=1}^{m_l+1}\sum_{j=s-1}^{m_l} t_{\zeta_l,j}(x-\lambda_l)^{-j+s-2}s \oint_{\mathcal{B}_{\zeta_l,s}}\omega_{0,2}(\cdot,z') \\
& = \sum_{s=1}^{m_l+1}\sum_{j=s-1}^{m_l} \frac{t_{\zeta_l,j}(x-\lambda_l)^{s-1}}{(x-\lambda_l)^{-j+1}} \underset{p_1=\zeta_l}{\mathrm{Res}} \frac{\omega_{0,2}(p_1,z')}{(x(p_1)-\lambda_l)^s}= \underset{p_1=\zeta_l}{\mathrm{Res}} \sum_{j=0}^{m_l} \frac{t_{\zeta_l,j}\, B(p_1,z')}{(x-\lambda_l)^{j+1}}\sum_{s=1}^{j+1}\frac{(x-\lambda_l)^{s-1}}{(x_1-\lambda_l)^s} ,
\end{align*}
where we denoted $x_1=x(p_1)$ and $\omega_{0,2}=B$ is the Bergman kernel. Developing the last sum and dropping the terms which are regular at $p_1\rightarrow\zeta_l$, we obtain
\beq\label{irreg1}
\underset{p_1=\zeta_l}{\mathrm{Res}} \frac{y(p_1)}{x_1-x}B(p_1,z'),
\eeq
which is indeed irregular when $z'\rightarrow \zeta_l$. Observe that this is the local behavior around $z'=\zeta_l$ of the derivative of the second term of the RHS of \eqref{lastExp} with respect to $x'$. Now taking the derivative of the first term as well, we get
\beq\label{irreg2}
\frac{\dd}{\dd x'}\frac{y(x')}{x-x'}=\underset{p_1=z'}{\mathrm{Res}} \frac{y(p_1)}{x-x_1}B(z',p_1).
\eeq
Finally, adding \eqref{irreg1} and \eqref{irreg2}, we obtain an expression which is regular at $z'\rightarrow\zeta_l$:
$$
\frac{1}{2\pi i}\oint_{\gamma}\frac{y(\cdot)}{x-x_1}B(z',\cdot),
$$
where $\gamma$ is a contour around $z'$ and $\zeta_l$.
\end{itemize}
Thus, \eqref{lastExp} is a rational function of $x$ and $x'$ with no poles at $x'\rightarrow\infty$ or $x'\rightarrow\lambda_l$.
\end{proof}

\begin{proof}[Proof of Theorem~\ref{thmApp}]
Using the shape of Lemma~\ref{lemmashape0}, we deduce that $\mathcal{O}\,\mathbb{K}(x,x')=0$ when $x'\rightarrow x$ and hence also the expression given by \eqref{eqshape1}~vanishes when $x'\rightarrow x$. In Lemma~\ref{lemmashape1}, we have shown that \eqref{eqshape1} has no poles as a rational function of $x'$. Thus it must be constant, and since it vanishes at $x'\rightarrow x$, it must be zero, which implies $\mathcal{O}\,\mathbb{K}(x,x')=0$.
\end{proof}

\bibliographystyle{plain}
\bibliography{BibliPainleve}

\end{document}